\documentclass[aip, jmp,
amsmath,amssymb,
preprint,11pt
]{revtex4-1}

\usepackage{graphicx}
\usepackage{dcolumn}
\usepackage{bm}

\usepackage[utf8]{inputenc}
\usepackage[T1]{fontenc}
\usepackage{mathptmx}
\usepackage{etoolbox}

\usepackage{amsthm}

\newtheorem{theorem}{Theorem}[section]
\newtheorem{corollary}{Corollary}[theorem]
\newtheorem{proposition}{Proposition}[section]
\newtheorem{lemma}{Lemma}[section]
\newtheorem{example}{Example}[section]
\newtheorem{remark}{Remark}[section]

\DeclareMathOperator{\tr}{tr}

\makeatletter
\def\@email#1#2{%
 \endgroup
 \patchcmd{\titleblock@produce}
  {\frontmatter@RRAPformat}
  {\frontmatter@RRAPformat{\produce@RRAP{*#1\href{mailto:#2}{#2}}}\frontmatter@RRAPformat}
  {}{}
}%
\makeatother

\begin{document}


\title[Schwarz maps with symmetry]{Schwarz maps with symmetry}

\author{A. Garc\'{\i}a-Velo}
\affiliation{Depto. de Matem\'aticas, Univ. Carlos III de Madrid, \\ Avda. de la Universidad 30, 28911 Legan\'es, Madrid, Spain.}
\email{alfonsog@fis.uc3m.es}

\author{A. Ibort}
\affiliation{Instituto de Ciencias
Matem\'{a}ticas (CSIC - UAM - UC3M - UCM) ICMAT and Depto. de Matem\'aticas, Univ. Carlos III de Madrid, \\ Avda. de la
Universidad 30, 28911 Legan\'es, Madrid, Spain.}
\email{albertoi@math.uc3m.es}

\date{\today}

\begin{abstract}
The theory of symmetry of quantum mechanical systems is applied to study the structure and properties of several classes of relevant maps in quantum information theory: CPTP, PPT and Schwarz maps. First, we develop the general structure that equivariant maps $\Phi:\mathcal A \to \mathcal B$ between $C^\ast$-algebras satisfy. Then, we undertake a systematic study of unital, Hermiticity-preserving maps that are equivariant under natural unitary group actions. Schwarz maps satisfy Kadison's inequality $\Phi(X^\ast X) \geq \Phi(X)^\ast \Phi(X)$ and form an intermediate class between positive and completely positive maps. We completely classify $U(n)$-equivariant on $M_n(\mathbb C)$ and determine those that are completely positive and Schwarz. Partial classifications are then obtained for the weaker $DU(n)$-equivariance (diagonal unitary symmetry) and for tensor-product symmetries $U(n_1) \otimes U(n_2)$. In each case, the parameter regions where $\Phi$ is Schwarz or completely positive are described by explicit algebraic inequalities, and their geometry is illustrated. Finally, we further show that the $U(n)$-equivariant family satisfies $\mathrm{PPT} \iff \mathrm{EB}$, while the $DU(2)$, symmetric $DU(3)$, $U(2) \otimes U(2)$ and $U(2) \otimes U(3)$, families obey the $\mathrm{PPT}^2$ conjecture through a direct symmetry argument. These results reveal how group symmetry controls the structure of non-completely positive maps and provide new concrete examples where the $\mathrm{PPT}^2$ property holds.
\end{abstract}

\maketitle

\section{Introduction}
\label{Section: Introduction}

The rapid development of quantum information theory is leaving a gap in our understanding of the role played by symmetry. Indeed, as information is physical \cite{La91}, the symmetries present on the physical background will emerge on the properties and characteristics of information systems. Symmetry proves itself valuable across various domains of quantum information theory, including applications in quantum cloning \cite{Keyl_99}, programmable devices \cite{Gschwendtner_21}, and the quantum extension of Shannon's theory \cite{Werner_02, King_03, Holevo_05, Datta_06}, among others.

The theory of information channels (both classical and quantum), i.e., of various classes of linear maps $\Phi:\mathcal A \to \mathcal B$ between $C^\ast$-algebras, is a rich field with many open problems. Indeed, the parametrization and classification of quantum channels is a central problem in quantum information theory, in particular those that exhibit symmetry under a group action. The role of symmetry in this context is still being articulated, especially concerning operational tasks \cite{Bartlett_07}, resource theories \cite{Gour_08}, and the structural analysis of dynamical processes \cite{Chiribella_08}.

The study of quantum channels with symmetry was initiated by Scutaru \cite{Scutaru_79} and developed by Holevo \cite{Holevo_93, Holevo_96}. If symmetry is present, the description of quantum channels can be notably simplified (see, for instance, recent contributions \cite{Nuwairan_14, Mozrzymas_17, Siudzinska_18, Memarzadeh_22, Lee_22, Vollbrecht_01, Collins_2018, Bardet_20, Si21, Park_24, Si22} and references therein). For instance, symmetry has helped making tractable various problems related to quantum entanglement \cite{Vollbrecht_01, Collins_2018, Christandl_19, Bardet_20, Si21, Park_24, Si22}.

The theory of positive and completely positive maps, although well-established in specific cases, remains incomplete. Key questions regarding their structure \cite{Paulsen_03}, extremal properties \cite{Choi_75}, and their relation with symmetry are still under active investigation. In infinite dimensions, these challenges are compounded by additional technical subtleties, making the development of symmetry-based tools not only elegant but necessary.

In this work, we will investigate the general form that all information channels which present some symmetry must satisfy. Once the general structure results for symmetric maps are established, we show its usefulness by applying it to study the interplay between symmetry and Schwarz maps, quantum channels and the PPT$^2$ conjecture. Schwarz maps define a family of positive, but in general non-completely positive, maps that satisfying the so called Kadison's inequality \cite{Kadison_52, Choi_74, CHOI_80}, (see Sect. \ref{Subsection: Positive, Schwarz and completely positive maps}). Schwarz maps provide a class of relevant entanglement witnesses \cite{Chruscinski_14}, but have proved to be extremely difficult to characterize \cite{Chruscinski_19, Chruscinski_21, Carlen_23, Chruscinski_24_1, Ch24}. The structure and properties of maps invariant under the full unitary group $U(n)$ will be completely elucidated, showing that, even in this simple scenario, there is a family of Schwarz maps which are not completely positive. In addition to this, it will be shown that $U(n)$-equivariant PPT maps are EB (see Sects. \ref{Subsection: Entanglement breaking channels: the PPT$^2$ conjecture}, for more details).

When the symmetry group is smaller, the classification and characterization of channels possessing this symmetry becomes increasingly difficult. We will provide several results describing structural properties of $U(n_1) \otimes U(n_2)$ and $DU(n)$-equivariant channels. Both families of maps are important in quantum information theory since they include several relevant channels such as the depolarizing and transpose depolarizing channels, amplitude damping channels, Schur multipliers, and others \cite{Si22b}. Diagonal-unitary channels have received a lot of attention \cite{Mozrzymas_17, Si21, Si22} because they satisfy the PPT$^2$ conjecture \cite{Si22b}. The positivity, complete positivity, and Schwarz properties of $DU(2)$ invariant channels were analyzed and completely determined in \cite{Chruscinski_24_1, Ch24}. Such efforts will be continued here by addressing the structural properties of Schwarz $DU(n)$-equivariant channels, obtaining a complete classification for a class of $DU(3)$-equivariant channels. Moreover, a direct proof of the PPT$^2$ conjecture will be provided in this case. Finally, we will address the product symmetry $U(n_1) \otimes U(n_2)$. Again, the structure of Schwarz, CP and the PPT$^2$ conjecture will be analyzed in small dimensions $n_1 = 2$, $n_2 = 2,3$.

The paper will be organized as follows: Sect. \ref{Section: Analyzing quantum channels} will establish the basic notions and notations on the theory of quantum channels used in the paper. This will be followed by a succinct review of the theory of group representations, Sect. \ref{Section: Fundamentals of representation theory}, that will serve as the natural language for encoding symmetry in quantum systems. The main theorem characterizing maps with symmetry will be established in Sect. \ref{Section: Channels with symmetry}, and some relevant instances are discussed in detail in Sect. \ref{Section: Parametrizing quantum equivariant maps}. More concretely, $U(n)$-equivariant channels are completely described and $DU(n)$-equivariant and $U(n_1) \otimes U(n_2)$-equivariant channels are also considered. Then we will concentrate in the theory of equivariant Schwarz maps. As a consequence of this work, a complete classification of $U(n)$-equivariant Schwarz maps will be established, see Theorem \ref{Theorem: Schwarz U-equivariant maps}, providing the first examples of Schwarz maps which are not CP in any dimension. Partial results for $U(n_1) \otimes U(n_2)$ and $DU(n)$ equivariant maps will also be presented in Sects. \ref{Subsection: Diagonal unitary equivariance}, \ref{Subsection: product unitary equivariance} as well as direct proofs of the PPT$^2$ conjecture in various instances.

\section{Analyzing quantum channels}
\label{Section: Analyzing quantum channels}

\subsection{Positive, Schwarz and completely positive maps}
\label{Subsection: Positive, Schwarz and completely positive maps}

Many efforts and ideas have been introduced to understand quantum channels like Stinespring structure theorem or the Choi-Kraus form of CPTP maps. One relevant result that will be used in what follows is the Choi-Jamiolkowski isomorphism. It defines an isomorphism between linear maps $\Phi: \mathcal A \to \mathcal B$, and operators $C_\Phi = \sum_{i,j} E_{ij} \otimes \Phi(E_{ij}) \in \mathcal A \otimes \mathcal B$, for every finite-dimensional algebras $\mathcal A, \mathcal B$ \cite{Choi_75}. The operators $E_{ij}$ define the canonical basis of $\mathcal A$, and $C_\Phi$ is called the Choi matrix of $\Phi$. Choi's theorem shows that a linear map $\Phi$ is completely positive (CP) if and only if its Choi matrix $C_\Phi$ is positive semi-definite \cite{Choi_75}. Similarly, $\Phi$ is positive if and only if $C_\Phi$ is block positive, that is $\langle v \otimes w, \, C_\Phi \, v \otimes w \rangle \geq 0$, for all $v \in \mathbb{C}^n$, $w \in \mathbb{C}^m$. It is clear that if $\Phi$ is completely positive then it is also positive. Thus, the set $\mathsf{CP}$ of completely positive maps is a strict subset of the set $\mathsf{P}$ of positive maps: $\mathsf{CP} \subset \mathsf{P}$.

We introduce now the class of positive maps relevant for this work. A unital linear map is said to be Schwarz if it satisfies Kadison's inequality:
\begin{equation}
\label{Equation: Schwarz inequality}
    \Phi(X^\dagger X) \geq \Phi(X)^\dagger \Phi(X), \quad \forall \, X \in \mathcal A \, .
\end{equation}

For future use we denote
\begin{equation}\label{eq:MPhi}
M_\Phi(X) \equiv \Phi(X^\dagger X) - \Phi(X)^\dagger \Phi(X) \, .
\end{equation}

It turns out that being Schwarz is a condition that lies in between being positive and completely positive  \cite{Paulsen_03, Stormer_12, Paulsen_07, Chruscinski_21}. That is $\mathsf{CP} \subset \mathsf{S} \subset \mathsf{P}$, where $\mathsf{S}$ denotes the set of Schwarz maps and all inclusions are strict. The first example of a non-completely positive Schwarz map was given by Choi \cite{CHOI_80}. Interestingly, Kadison's inequality is satisfied for normal operators \cite{Choi_74, CHOI_80}. As indicated in the introduction, some examples of the usefulness of the Schwarz inequality are related to the characterization of quantum Markovian evolution \cite{Chruscinski_19}, the asymptotic dynamics of open quantum systems \cite{Amato_23, Amato_24}, and the proof of monotonicity results \cite{Hiai_12, Carlen_22, Carlen_23}.

\subsection{Entanglement breaking channels: the PPT$^2$ conjecture}
\label{Subsection: Entanglement breaking channels: the PPT$^2$ conjecture}

Determining whether or not a state is entangled is a NP-hard problem. Horodecki's theorem \cite[Thm. 2]{Horodecki_96}, states that a bipartite quantum state $\rho \in \mathcal A \otimes \mathcal B$ is separable if and only if $\mathrm{id}_{\mathcal A} \otimes \Phi (\rho) \geq 0$ for all positive maps $\Phi$. The typical positive but non-completely positive map used is the transpose map $T$ \cite{Peres_96}. Similar criteria has been obtained by generalizing the concept to a higher number of subsystems \cite{Horodecki_01}. Because of this $\mathsf{P}$ but not $\mathsf{CP}$ maps were recognized as a powerful tool for entanglement detection \cite{Chruscinski_14}.

The modern approach to entanglement is to consider it as a resource. Correlations can be created and consumed for more efficient information tasks \cite{Chitambar_19}. The use of this resource requires to act on states that carry entanglement using quantum channels, for instance modeling their transfer from one point to another, or the interactions required for the task at hand.

An open problem in this regard is the study of entanglement breaking (EB) channels. As its name suggests, an entanglement breaking channel is a channel $\Phi$ such that $\mathrm{id}_k \otimes \Phi (\rho)$ is separable for every $k \in \mathbb{N}$ and $\rho \in M_k(\mathbb C) \otimes \mathcal A$ \cite{Horodecki_03}. A channel $\Phi$ is EB if and only if its Choi matrix $C_\Phi$ is separable \cite{Horodecki_03}. Thus it is again a NP-hard problem. From the study of EB channels the PPT$^2$ conjecture arose \cite{Christandl_12}. A quantum channel $\Phi$ is said to be PPT (positive partial transpose) if it satisfies $T \circ \Phi \geq 0$. The conjecture states that the channel formed by the serial composition of two PPT quantum channels is always EB. That is, if $\Phi$ is PPT, then $\Phi \circ \Phi \equiv \Phi^2$ is EB.

One of its main consequences is that, in the case of being true, any PPT-entangled state (an state with positive partial transpose $T \otimes \mathrm{id}_{\mathcal B} (\rho) \geq 0$) cannot be used as a quantum resource in quantum networks \cite{Christandl_17}. The difficulty of characterizing EB channels is translated here. The PPT$^2$ conjecture has been proven for channels restricted to certain conditions, such as low Kraus rank channels, approximately depolarizing channels, Gaussian channels, diagonal unitary equivariant, and random channels \cite{Lami_15, Kennedy_18, Rahaman_18, Collins_18_2, Hanson_20, Si22}. Furthermore it holds for systems of dimension $2$ and $3$ \cite{Chen_19}.

\section{Fundamentals of representation theory}
\label{Section: Fundamentals of representation theory}

In this section the basic ingredients of symmetry and the theory of representations of groups that support it will be described (see, for instance \cite{Ibort_20} and references therein).

\subsection{Groups, group actions, and symmetry}
\label{Subsection: Groups, group actions, and symmetry}

We will restrict ourselves to consider finite or compact groups. To fix notation we recall that a group is a set $G$ with a composition law $\circ: G \times G \to G$, satisfying thee properties. i) Associativity: for all $g, g', g'' \in G$, the composition satisfies $g \circ (g' \circ g'') = (g \circ g') \circ g''$. ii) Neutral element: There exists an element $e \in G$, such that $e \circ g = g \circ e = g$ for all $g \in G$. Finally, iii) Inverse element: for all $g \in G$, there exits an element $g^{-1} \in G$, such that $g \circ g^{-1} = g^{-1} \circ g = e$.  If the composition law of the group is commutative, that is, $g \circ g' = g' \circ g$, for all $g, g' \in G$, we say that the group is called Abelian.

But how does it relate to the concept of symmetry? Consider a physical system, which is described mathematically by using some supporting or carrier space $V$. The carrier space can be just a set of objects, for instance, in the context of classical communications, and alphabet. Although usually it is equipped with additional structures like a topological or an algebraic structure. In the applications to be considered in this work $V$ is a topological vector space, and more concretely, either a Hilbert space or a $C^\ast$-algebra.

The group of symmetries $G$ will be realized on the carrier space $V$. To materialize that, one defines the action of $G$ on $V$, which is a map:
\begin{eqnarray}
    \Psi \colon G \times V & \to & V \\
    (g, v) & \to & \Psi(g, v) \equiv \Psi_g (v)\, , \nonumber
\end{eqnarray}
satisfying: $\Psi_e (v) = v$ and $\Psi_{g \circ g'} (v) = \Psi_g \circ \Psi_{g'} (v)$ for all $v \in V$ and $g \in G$. Moreover, if the group $G$ carries a topology, we impose the map $\Psi_v \colon G \to V$, $\psi_v(g) := \psi_g (v)$, to be continuous. In other words, the group $G$ is realized by means of transformations $\Psi_g$ in the carrier space $V$ (note that the maps $\Psi_g$ are all bijections).

If $V$ possesses any additional structure of interest, we will impose $\Psi$ to preserve such structure. For instance, we define $\Psi_g$ to be linear due to the fact that $V$ is assumed to be a linear space. In this sense a `symmetry' is the preservation of some structure on $V$. The `group of symmetries' of a given structure $S$ on $V$ is a group $G$ acting on $V$ whose action preserves such structure.

\begin{example}
    Let $\mathcal B(\mathcal H)$ be the $C^\ast$-algebra of all bounded operators on the Hilbert space $\mathcal H$. Then, the action of the unitary group $U(\mathcal H)$ is given by the map $\Psi_U (X) = U X U^\dagger$, for all $X \in \mathcal B(\mathcal H)$ and $U \in U(\mathcal H)$. Notice that $\Psi_{\mathrm{id}_V} (X) = X$ and $\Psi_{UV} (X) = U (V X V^\dagger) U^\dagger = \Psi_U \circ \Psi_V (X)$. Moreover the action is linear and continuous ($\mathcal B(\mathcal H)$ carrying, for instance the strong topology).
    
    We may consider the Frobenius inner product $\langle X, Y \rangle = \mathrm{Tr\,} X^\dagger Y$, which is defined on the space of Hilbert-Schmidt operators $\mathcal L^2(\mathcal H)$ on $\mathcal{H}$ that becomes a Hilbert space. Again unitary maps preserve this particular structure since $\langle \Psi_U (X), \Psi_U (Y) \rangle = \langle X, Y \rangle$, for all $U \in U(n)$. This implies that the representation $\Psi$ becomes a unitary representation of the unitary group $U(\mathcal H)$ with support the Hilbert space $\mathcal L^2(\mathcal H)$. If $\mathcal H$ is finite-dimensional, then the two situations agree as the $C^\ast$-algebra of bounded operators in $\mathcal H$ can be identified with the $C^\ast$-algebra of $n \times n$-matrices, $M_n(\mathbb C)$, which is itself a Hilbert space with the Frobenius inner product.
\end{example}

\subsection{Linear unitary representations of groups}
\label{Subsection: Linear unitary representations of groups}

The action of a group on a linear space by linear transformations is called a linear representation of the group. Its study and analysis provides a powerful tool for understanding how abstract groups manifest in the concrete spaces where the laws of physics are defined.

Thus a representation of $G$ into the linear vector space $V$ is no more than an association between the elements of $G$ and the elements of the group of linear isomorphisms on it, called the general linear group of $V$ and denoted as $GL(V)$. If the linear space $V$ is finite-dimensional, $\dim V = n$ say, then $GL(V)$ is isomorphic to the group $GL(n)$ of invertible square matrices of dimension $n$. Formally, a representation is denoted by the triplet $(G, V, \mu)$:
\begin{equation}
    \mu \colon G  \to  GL(V) \, , \qquad g  \mapsto  \mu (g) \equiv \mu_g: V \to V \, , 
\end{equation}
and $\mu$ is a group homomorphisms. What this means is that $\mu_{g \circ g'} = \mu_g \circ \mu_{g'}$, for all $g, g' \in G$, and the maps $g \in G \mapsto \mu (g) v \in V$, for all $v\in V$ are continuous. Because $\mu$ is a homomorphism it translates the structure of the group into $V$. Furthermore we get immediately that $\mu_e = \mathrm{id}_V$ and $\mu_{g^{-1}} = \mu_g^{-1}$, for all $g \in G$. The linear space $V$ is called the support of the representation, and we say that the dimension of the representation is the dimension of $V$: $\dim \mu := \dim V$.

If $V$ has additional structure, like a Hilbert space structure, it should be preserved by the representation. Thus if $V$ is a Hilbert space $\mathcal H$, we ask the linear representation $\mu$ to be such that $\mu(g)$ is a unitary map for all $g\in G$. Thus $\mu: G \to U(\mathcal H) \subset GL(\mathcal H)$, where $U(\mathcal H)$ denotes the group of unitary transformations on the space $\mathcal H$. We will call such linear representation a unitary representation.

If $V$ is a $C^\ast$-algebra $\mathcal A$, we should ask the representation $\mu$ to preserve the algebraic structure of $\mathcal A$, thus, $\mu_g (XY) = \mu_g(X) \mu_g(Y)$, and $\mu_g (X^\ast) = [\mu_g(X)]^\ast$, for all $X, Y \in \mathcal A$, $g \in G$. These conditions imply that $\mu(g)$ is a $\ast$-automorphism of the $C^\ast$-algebra $\mathcal A$ and that the representation $\mu$ is a group homomorphism $\mu \colon G \to \mathrm{Aut\, }(\mathcal A)$. A representation of a group on a $C^\ast$-algebra $\mathcal A$ is also known as a dynamical system on $\mathcal A$.

\subsection{Equivalent representations}
\label{Subsection: Equivalent representations}

We face the question of given a group $G$, how many different representations exist. So, in order to answer this question, we have to determine when two given representations have the same structural properties. We say that two representations $(V, \mu)$ and $(W, \nu)$, of the same group $G$, are equivalent, and denote it by $\mu \sim \nu$, if there exists an isomorphism $\phi: V \to W$ such that $\phi \circ \mu_g = \nu_g \circ \phi$, for all $g \in G$.

This definition fulfills the properties of an equivalent relation, and so the family of all representations of $G$ can be decomposed into equivalence classes of representations $[\mu] = \{ \nu \, | \, \mu \sim \nu \}$. We are interested not in the properties of individual representations, but in the properties of the representations in the same equivalence class.

If we are considering linear representations of compact groups on Hilbert spaces it is not hard to see that every linear representation $(V, \mu)$ of $G$, is equivalent to a unitary representation of $G$. Therefore, for compact groups we can work just with unitary representations. This is of special interest for quantum mechanics, since the support space will be a Hilbert space.

On the contrary if we were working with $C^\ast$-algebras, the analysis should be conducted in a different way. Of course, in the finite-dimensional situation, that is $\dim \mathcal A < \infty$, it can be proved that $\mathcal A$ is a finite direct sum of matrix algebras $M_{n_k}(\mathbb C)$, $n_1, \ldots, n_r \in \mathbb{N}$, and we are in the previous situation. The channels to be considered in this paper are going to be always defined on finite-dimensional spaces, however the infinite-dimensional context will be discussed elsewhere.

\subsection{Irreducible representations and Schur's lemma}
\label{Subsection: Irreducible representations and Schur's lemma}

Given that all representations in the same equivalence class have similar properties, we can ask ourselves what are the simplest representations that we can choose to work with. To find them we will decompose an arbitrary representation into simpler blocks that cannot be reduced further. For this argument we will be considering just unitary representations of a group $G$ on a Hilbert space $\mathcal H$.

Consider closed subspaces  $W \subset V$ that are left invariant by $\mu$, i.e., $\mu_g w \in W$ for all $w \in W$ and $g \in G$; we will call them $G$-invariant subspaces. If there exits a proper closed $G$-invariant subspace $W \subset V$, that is, different from $\{ 0 \}$ and $V$, the representation is called reducible. Otherwise, is said to be irreducible. The family of equivalence classes of irreducible unitary representations of the group $G$ will be denoted by $\widehat{G}$.

For every finite-dimensional unitary representation $\mu$, there exists a unique decomposition of $V$ into a direct sum of a finite number of $G$-invariant subspaces $W_i$ such that each subspace $W_i$ supports a unique irreducible representations $\mu_i$ of $G$:
\begin{equation} \label{Equation: Decomposition representations}
    V = W_1 \oplus \cdots \oplus W_r \, , \qquad \mu = \mu_1 \oplus \cdots \oplus \mu_r \, .
\end{equation}

Representations $\mu$ having this property are called completely reducible, and thus finite-dimensional unitary representations always are completely reducible. For infinite-dimensional representations and non-compact groups this is not necessarily true and the situation becomes much more complex (see \cite[Ch. 8]{Ki76}) for an overall discussion on the subject.

Given two representations $(G, W_1, \mu_1)$ and $(G, W_2, \mu_2)$, their direct sum is another linear representation $(G, W_1 \oplus W_2, \mu_1 \oplus \mu_2)$. The direct sum of linear vector spaces is defined as usual. If $w_{1, 2} \in W_{1, 2}$, then $[\mu_1 \oplus \mu_2] (g) (w_1 \oplus w_2) = \mu_1 (g) (w_1) \oplus \mu_2 (g) (w_2)$, for all $g \in G$. The same applies to the construction of the tensor product  $\mu_1 \otimes \mu_2$ of two representations on the space $W_1 \otimes W_2$.

Some of the irreducible representations $\mu_i$ of $G$ that appear in the decomposition (\ref{Equation: Decomposition representations}) may be equivalent. In such case we may rewrite the decomposition (\ref{Equation: Decomposition representations}) as:
\begin{equation} \label{Equation: Decomposition equivalent representations}
    V = \bigoplus_{\alpha \in \widehat{G}} \mathbb C^{n_\alpha} \otimes W^{(\alpha)}, \qquad \mu = \bigoplus_{\alpha \in \widehat{G}} \mathbb I_{n_\alpha} \otimes \mu^{(\alpha)} \, ,
\end{equation}
where $\alpha \in \widehat{G}$ labels the equivalence classes of irreducible representations of $\mu$ with representatives $\mu^{(\alpha)}$. Also notice that $\mathbb C^{n_\alpha} \otimes W^{(\alpha)} \cong W^{(\alpha)} \oplus \cdots \oplus W^{(\alpha)}$ and $\mathbb I_{n_\alpha} \otimes \mu^{(\alpha)} \cong \mu^{(\alpha)} \oplus \cdots \oplus \mu^{(\alpha)}$ where the sum is repeated $n_\alpha$ times. Thus, we have reordered the terms in (\ref{Equation: Decomposition representations}), to obtain an isomorphic space.

Finally we will state the main result concerning linear maps that intertwine two given representations:

\begin{lemma}[Schur's Lemma]
\label{Lemma: Schur's Lemma}
    Let $(G, V_1, \mu_1)$ and $(G, V_2, \mu_2)$ be two finite-dimensional irreducible representations, and let $\phi: V_1 \to V_2$ be a continuous linear map satisfying $\phi \circ \mu_1 (g) = \mu_2 (g) \circ \phi$, for all $g \in G$. Then:
    
    \begin{enumerate}
        \item If $V_1 \ncong V_2$, then $\phi = 0$.
        \item If $V_1 = V_2$, then $\phi = \lambda \, \mathrm{id}_V$, where $\lambda \in \mathbb C$. 
      \end{enumerate}
\end{lemma}

Again, there is an infinite-dimensional version of Schur's Lemma that will not be necessary in this work.

\section{Channels with symmetry}
\label{Section: Channels with symmetry}

After the detour to study how symmetry is represented in abstract Hilbert spaces and $C^\ast$-algebras we will apply it now to the study of the structure of quantum channels. The symmetry of a map $\Phi : \mathcal A \to \mathcal B$, where $\mathcal A$, $\mathcal B$ are either Hilbert spaces of $C^\ast$-algebras, manifests if the identity
\begin{equation}
    \Phi \circ \mu_g = \nu_h \circ \Phi \, , \qquad \forall \, g \in G \, , \ h \in H \, ,
\end{equation}

is satisfied, where $\mu$ and $\nu$ are representations of, in principle, two different groups $G$ and $H$ respectively. In this context, we say that the map is $(G, H)$-equivariant, $\mu$-equivariant, or just $G$-equivariant if $G = H$.   
A useful observation is that a map is $(G,H)$-equivariant if and only if its Choi matrix is $(\mu_G \otimes \bar{\nu}_G)$-invariant. With this, we can state now the main structure theorem for channels with symmetry.

\begin{theorem}[Channel decomposition]
\label{Theorem: Channel decomposition}
Let $G$ be a compact group and $\Phi \colon \mathcal A \to \mathcal B$ be a $G$-equivariant map with $\mathcal A$, $\mathcal B$ finite-dimensional representations of $G$. Then there exists a unique decomposition (up to reordering):
\begin{equation*}
    \Phi  = \bigoplus_{\alpha \in \widehat{G}} \mathcal C_\alpha \otimes \mathrm{id}_\alpha \, ,
\end{equation*}
where $\alpha \in \widehat{G}$ indexes the irreducible representations of $G$, $n_\alpha$ are their multiplicities, and $\mathcal C_\alpha$ are reduced channel components.
\end{theorem}

\begin{proof}
    Any irreducible representation of a compact group is finite-dimensional \cite{Ibort_20}. Moreover, because of Peter-Weyl theorem, the space of $\widehat G$ of irreducible representations of $G$ is countable \cite{Ibort_20}. Because we assume that $\mathcal A$, $\mathcal B$ are finite-dimensional, they admit a decomposition $\mathcal A = \bigoplus_{i} \mathcal A_i$ and $\mathcal B = \bigoplus_{j} \mathcal B_j$.
    
    Then, any linear map $\Phi \colon \mathcal A \to \mathcal B$ can be expressed as $\Phi  = \bigoplus_{i, j} \Phi _{ij}$. The linear map $\Phi _{ij}: \mathcal A_i \to \mathcal B_j$ is defined by assigning $X_i$ (the $i$-th component of $X = \bigoplus_i X_i$); to $Y_j$ (the $j$-th component of $Y = \Phi (X)$), i.e. $Y_j = \Phi_{ij}(X_i)$, which lies in $\mathcal B_j$.
    
    Because of (\ref{Equation: Decomposition equivalent representations}), since $G$ is compact, $\mathcal A$ and $\mathcal B$ decompose as
    \begin{equation*}
        \mathcal A = \bigoplus_{i = 1}^r \mathcal A_i \equiv \bigoplus_{\alpha} \mathbb C^{n_\alpha} \otimes V_\alpha \, , \qquad \mathcal B = \bigoplus_{j = 1}^r \mathcal B_j \equiv \bigoplus_{\alpha} \mathbb C^{n'_\alpha} \otimes W_\alpha \, ,
    \end{equation*}
    
    with $n_\alpha$ and $n'_\alpha$ the multiplicities of $V_\alpha$ and $W_\alpha$ respectively.
    
    Now enters into play the assumption that $\Phi $ is $G$-equivariant. For two representations $(\mathcal A, \mu)$ and $(\mathcal B, \nu)$ of $G$, then $\Phi  \circ \mu_g (X) = \nu_g \circ \Phi (X)$, witch implies equivariance of the reduced channel components:    
    \begin{equation*}
        \Phi _{ij} \circ \mu_g^{(i)} = \nu_g^{(j)} \circ \Phi _{ij} \, .
    \end{equation*}
    
    Finally, by Schur's lemma (Lem. \ref{Lemma: Schur's Lemma}), $\Phi _{ij}$ is block diagonal in this decomposition. If $\mu_i \sim \nu_j$ ($\mathcal A_i \cong \mathcal B_j$), then $\Phi _{ij} = \lambda \, \mathrm{id}_i$ where $\lambda \in \mathbb C$; and $\Phi _{ij} = 0$ otherwise. Thus, rearranging the terms:    
    \begin{equation}
        \Phi  = \bigoplus_{\alpha \in \widehat{G}} \mathcal C_\alpha \otimes \mathrm{id}_\alpha \, , \qquad \mathcal C_\alpha: \mathbb C^{n_\alpha} \to \mathbb C^{n'_\beta} \, .
    \end{equation}
\end{proof}

\begin{lemma}[Composition of $G$-equivariant channels]
\label{Lemma: Composition of $G$-equivariant channels}
Let $\Phi \colon \mathcal A \to \mathcal A'$ and $\Psi \colon \mathcal A' \to \mathcal A''$ be two linear maps between finite dimensional $C^\ast$-algebras. Let these maps be $G$-equivariant with decompositions $\Phi = \bigoplus_{\alpha \in \widehat{G}} \mathcal C_\alpha^{\Phi} \otimes \mathrm{id}_\alpha$ and $\Psi = \bigoplus_{\alpha \in \widehat{G}} \mathcal C_\alpha^{\Psi} \otimes \mathrm{id}_\alpha$ respectively. Then the composition map $\Psi \circ \Phi$ is $G$-equivariant, and reads
\begin{equation*}
    \Psi \circ \Phi = \bigoplus_{\alpha \in \widehat{G}} \left( \mathcal C^{\Psi}_\alpha \circ \mathcal C^{\Phi}_\alpha \right) \otimes \mathrm{id}_\alpha \, .
\end{equation*}

If $\Gamma^{\Phi}_\alpha$ and $\Gamma^{\Psi}_\alpha$ are the matrices representing the maps $\mathcal C^{\Phi}_\alpha$ and $\mathcal C^{\Psi}_\alpha$, then $\mathcal C^{\Phi}_\alpha \circ \mathcal C^{\Psi}_\alpha \cong \Gamma^{\Phi}_\alpha\Gamma^{\Psi}_\alpha$ as the usual product of matrices.
\end{lemma}

\begin{proof}
    Let $X \in \mathcal A$, then it can be written as $X = X_1 \oplus \cdots \oplus X_r \cong \bigoplus_\alpha v_x^{(\alpha)} \otimes X_\alpha$, where $X_\alpha \in W_\alpha$ and $v_x^{\alpha} \in \mathbb C^{n_\alpha}$. Then, by Thm.\ref{Theorem: Channel decomposition}, the map $\Phi$ acts independently in each term
    \begin{equation*}
        \Phi(X) = \bigoplus_\alpha \mathcal C^{\Phi}_\alpha \left(v_x^{(\alpha)}\right) \otimes X_\alpha \, .
    \end{equation*}

    The output $\Phi(X) \in \mathcal A'$, with $\mathcal C^{\Phi}_\alpha \left(v_x^{(\alpha)}\right) \in \mathbb C^{n_\alpha'}$. If we now act with $\Psi$,
    \begin{equation*}
        (\Psi \circ \Phi)(X) = \bigoplus_\alpha \left( \mathcal C^{\Psi}_\alpha \circ \mathcal C^{\Phi}_\alpha \right) \left(v_x^{(\alpha)}\right) \otimes X_\alpha \, .
    \end{equation*}

    Since $\mathcal C^{\Phi}_\alpha$ and $\mathcal C^{\Psi}_\alpha$ are linear maps between linear vector spaces, given orthonormal bases of those spaces, we can stablish the equivalences $\mathcal C^{\Phi}_\alpha \cong \Gamma^{\Phi}_\alpha$ and $\mathcal C^{\Psi}_\alpha \cong \Gamma^{\Psi}_\alpha$, where $\Gamma^{\Phi}_\alpha \in \mathbb C ^{n_\alpha \times n_\alpha'}$ and $\Gamma^{\Psi}_\alpha \in \mathbb C^{n_\alpha' \times n_\alpha''}$. Thus, the composition of maps just represents the usual matrix multiplication $\mathcal C^{\Psi}_\alpha \circ \mathcal C^{\Phi}_\alpha \cong \Gamma_\alpha^{(\Psi \circ \Phi)} = \Gamma_\alpha^\Psi \Gamma_\alpha^\Phi$.
\end{proof}

\begin{corollary}
    Let $\Phi \colon \mathcal A \to \mathcal B$ be a linear map between two finite dimensional $C^\ast$-algebras. If $\Phi$ is $G$-equivariant, then
    \begin{equation*}
        \Phi^k = \bigoplus_{\alpha \in \widehat G} \mathcal C_\alpha^k \otimes \mathrm id_\alpha \, ,
    \end{equation*}

    and it is parametrized by $|\widehat G|$ matrices $\Gamma_\alpha^k \in \mathbb C^{n_\alpha \times n_\alpha'}$.
\end{corollary}

As a bonus we can complete the previous structure result with the analysis of the capacity of a $G$-equivariant channel.

\begin{theorem}[Capacity decomposition]
\label{Theorem: Capacity decomposition}
Let $\Phi \colon \mathcal A \to \mathcal B$ be a linear map between two finite dimensional $C^\ast$-algebras. If $\Phi$ is $G$-equivariant, then the classical capacity can be expressed as:
\begin{equation*}
    C(\Phi ) = \max_\alpha \{n_\alpha \, C_\alpha(\Phi)\} \, ,
\end{equation*} 

where $C_\alpha(\Phi)$ is the capacity of the reduced channel $C_\alpha$.
\end{theorem}

\begin{proof}
    Let $\{p_k, \rho_k\}$ be any ensemble decomposition achieving the Holevo capacity. The Holevo information is \cite{Holevo_77}:    
    \begin{equation*}
        \chi (\{p_k, \rho_k\}) = S \left(\left[\sum_k p_k \Phi  (\rho_k)\right]\right) - \sum_k p_k S [\Phi (\rho_k)] \, .
    \end{equation*}
    
    By theorem \ref{Theorem: Channel decomposition}, each term $\Phi (\rho_k)$ decomposes as $\Phi (\rho_k) = \bigoplus_i c_i \, \rho_k^{(i)}$, where $\rho_k^{(i)} \in \mathcal A_i$. Since the von Neumann entropy is additive over direct sums,    
    \begin{equation*}
        S[\Phi (\rho_k)] = S \left[\bigoplus_i c_i \, \rho_k^{(i)}\right] = \sum_i S \left[c_i \, \rho_k^{(i)}\right] = \sum_\alpha n_\alpha S \left(\Phi \left[\rho_k^{(\alpha)} \right]\right) \, .
    \end{equation*}
    
    Then the Holevo information becomes    
    \begin{equation*}
    \chi \left( \{p_k, \rho_k\}\right) = \sum_\alpha n_\alpha \Phi_i \left(\left\{p_k, \rho_k^{(\alpha)}\right\}\right)\, .
    \end{equation*}
    
    Denoting $\chi \left(\left\{p_k, \rho_k^{(\alpha)}\right\}\right) \equiv \chi_\alpha$, the maximization can be performed independently for each $\alpha$, and  the final result follows by defining    
    \begin{equation}
        C_\alpha(\mathcal C) = \max_{\left\{p_k,\,\, \rho_k^{(\alpha)}\right\}} \chi_\alpha \, .
    \end{equation}
\end{proof}

\section{Equivariant CP and Schwarz maps}
\label{Section: Parametrizing quantum equivariant maps}

In this section we apply the previous results to some particular situations. We choose to analyze the structure of quantum maps equivariant with respect to different groups of interest \cite{Prudhoe_25}. For simplicity, we will assume the systems are defined by the finite-dimensional $C^\ast$-algebras $\mathcal A \cong \mathcal B \cong M_n (\mathbb C)$, and $\Phi \colon M_n (\mathbb C) \to M_n (\mathbb C)$ is a linear map. 

\subsection{Unitary equivariance}
\label{Subsection: Unitary equivariance}

One of the most natural, yet restrictive symmetries that one can consider is full unitary symmetry. The unitary group $U(n)$ of order $n$ is the set $U(n) = \{ U \in GL(n) \ | \ U U^\dagger = U^\dagger U = \mathbb I_n \}$.  The natural representation of $U(n)$ on $M_n (\mathbb C)$ is defined by conjugation:
\begin{equation}
    \mu: U(n)  \to  GL (M_n (\mathbb C) )  \, , \qquad   U  \mapsto  \mu_U(\cdot) = U \cdot U^\dagger\, . \nonumber
\end{equation}

Moreover, any automorphism $\beta$ of $M_n (\mathbb C)$ is inner, that is, there exists $U \in U(n)$ such that $\beta (X) = U X U^\dagger$. We conclude that for any $\ast$-representation $\mu \colon U(n) \to \mathrm{Aut\,}[M_n (\mathbb C)]$ there exists a family of unitary operators $U_g$ such that $\mu_g(X) = U_g X U_g^\dagger$. The operators $U_g$ do not have to define a unitary representation of $G$ though, however if $G$ is the group $U(n)$, our case, this is true.
\begin{proposition} [$U(n)$-equivariant maps]
\label{Proposition: U-equivariant maps}
    Let $\Phi \colon M_n (\mathbb C) \to M_n (\mathbb C)$ be a linear $U(n)$-equivariant map. Then,
    \begin{equation*}
        \Phi(X) = \frac{\sigma - \lambda}{n} \tr \, X \ \mathbb I_n + \lambda X \, ,
    \end{equation*}
    for all $X \in M_n (\mathbb C)$ and $\lambda, \sigma \in \mathbb{C}$.
\end{proposition}

\begin{proof}
    It is clear that the subspace $W_D \equiv \mathbb C \mathbb I_n := \{c \, \mathbb I_n\ | \ c \in \mathbb C\}$ is $U(n)$-invariant. Moreover, $W_D$ is irreducible since $\mu_U (c \, \mathbb I_n) = c \, \mathbb I_n$. The irreducible representation $\mu_D = \mu \mid_{W_D}$ is the trivial representation, since $\dim \mathbb C \mathbb I_n = 1$ and $\mu_U = 1$ for all $U \in U(n)$.
    
    Since $M_n (\mathbb C)$ is a Hilbert space with respect to the Frobenius inner product and $U(n)$ is represented unitarily on it, we can consider the orthogonal invariant subspace: $W_D^\perp = \{X \in M_n (\mathbb C) \ | \ \langle Y, X \rangle = 0, \ Y \in W_D\}$, which is given by the set of traceless matrices $W_D^\perp \equiv \mathfrak sl_n (\mathbb C) = \{X \in M_n (\mathbb C) \ | \ \tr \, X = 0\}$. It can be shown easily that $\mathfrak sl_n (\mathbb C)$ is irreducible by acting with $\mu$ on the elements of the canonical basis.    
    The decomposition of $\mu$ in irreducible representations reads    
    \begin{equation}
    \label{Equation: U decomposition}
        M_n (\mathbb C) \cong \mathbb C \mathbb I_n \oplus \mathfrak sl_n (\mathbb C) \, , \quad \mu \cong \mu_D \oplus \mu_D^\perp \, ,
    \end{equation}
    and $\dim \mathfrak{sl}_n (\mathbb C) = n^2 - 1$. By theorem \ref{Theorem: Channel decomposition}, $\Phi$ admits the decomposition    
    \begin{equation}
    \Phi = \sigma \, \mathrm{id}_D + \lambda \, \mathrm{id}_{D^\perp} \, , \qquad \sigma, \lambda \in \mathbb C \, .
    \end{equation}
    
    To understand its action on a general matrix $X \in M_n (\mathbb C)$, we need to decompose $X$ into a multiple of the identity, and a traceless part  
    \begin{equation*}
        X = \frac{\tr \, X}{n} \mathbb I_n + \left(X - \frac{\tr \, X}{n} \mathbb I_n \right) \, , \quad \forall \, \, X \in M_n (\mathbb C) \, .
    \end{equation*}
    
    Act with $\Phi$ on $X$ and the result follows.
\end{proof}

Thus the parametrization of any $U(n)$-invariant map, requires only 2 complex parameters $\lambda, \sigma$ for $n > 1$.  In what follows we impose further conditions on the map $\Phi$.

\begin{lemma}[Unital $U(n)$-equivariant channels]
\label{Lemma: Unital U-eqivariant channels}
    Let $\Phi \colon  M_n (\mathbb C) \to M_n (\mathbb C)$ be $U(n)$-equivariant. Then $\Phi$ is unital if and only if $\sigma = 1$, i.e.    
    \begin{equation}
    \Phi(X) = \frac{1 - \lambda}{n} \tr \, X \ \mathbb I_n + \lambda \, X, \quad \forall \, X \in M_n (\mathbb C) \, .
    \end{equation}
\end{lemma}
\begin{proof}
    $\Phi$ is unital if $\Phi(\mathbb I_n) = \sigma \, \mathbb I_n = \mathbb I_n$ which occurs if and only if $\sigma = 1$.
\end{proof}

Note that if $n = 1$, then $\Phi (X) = \sigma \, X$, for $\sigma \in \mathbb C$. Thus, by lemma \ref{Lemma: Unital U-eqivariant channels}, the map is unital if and only if $\Phi$ is the identity. The identity map is positive, Schwarz and CP; and thus the problem is tibial. In what follows we will assume without further mentioning it that $n > 1$.

\begin{lemma} [Hermiticity preserving $U(n)$-equivariant maps]
\label{Lemma: Hermiticity preserving U-equivariant maps}
    Let $\Phi \colon M_n (\mathbb C) \to M_n (\mathbb C)$ be $U(n)$-equivariant. Then $\Phi$ is Hermiticity preserving if and only if $\sigma, \lambda \in \mathbb R$.
\end{lemma}

\begin{proof}
    $\Phi$ is Hermiticity preserving if $\Phi\left(X^\dagger\right) = \Phi(X)^\dagger$. Then    
    \begin{equation*}
        \Phi \left(X^\dagger\right) - \Phi(X)^\dagger = \frac{(\sigma - \bar \sigma) - (\lambda  - \bar \lambda)}{n} \tr \, X^\dagger \, \mathbb I_n + (\lambda - \bar \lambda) \, X^\dagger = 0 \, ,
    \end{equation*}   
    which occurs for all $X$ if and only if $\lambda - \bar \lambda = 0$ and $(\sigma - \bar \sigma) - (\lambda  - \bar \lambda ) = \sigma - \bar \sigma = 0$, that is, $\sigma, \lambda \in \mathbb R$.
\end{proof}

We can now characterize completely $U(n)$-equivariant Schwarz maps $\Phi \colon  M_n (\mathbb C) \to M_n (\mathbb C)$.

\begin{theorem} [Schwarz $U(n)$-equivariant maps]
\label{Theorem: Schwarz U-equivariant maps}
    Let $\Phi \colon  M_n (\mathbb C) \to M_n (\mathbb C)$ be a $U (n)$-equivariant unital map. Then $\Phi$ is Schwarz if and only if $\lambda \in \left[ -\frac{1}{n}, 1 \right]$.
\end{theorem}

\begin{proof}
    To get necessary conditions, we evaluate the map at the elements $X = E_{ij}$ of the standard basis of $M_n(\mathbb{C})$. Then, test the positive semidefiniteness of $M_\Phi(X)$.

    \begin{itemize}
        \item If $i \neq j$: $\Phi(X) = \lambda E_{ij}$, $\Phi(X)^\dagger = \bar \lambda E_{ji}$ and $\Phi \left(X^\dagger X\right) = \frac{1 - \lambda}{n} \mathbb I_n + \lambda E_{jj}$. Then,        
        \begin{equation*}
            M_\Phi (E_{ij}) = \Phi \left(E_{ij}^\dagger E_{ij}\right) -  \Phi (E_{ij})^\dagger \Phi (E_{ij}) = \frac{1 - \lambda}{n} \mathbb I_n + (\lambda - |\lambda|^2) E_{jj} \geq 0 \, .
        \end{equation*}
        
        The matrix $M_\Phi (E_{ij})$ is diagonal, and thus has eigenvalues $\frac{1 - \lambda}{n}$ with multiplicity $n - 1$ and $\frac{1 - \lambda}{n} + (\lambda - |\lambda|^2)$ with multiplicity $1$. Notice that, by definition $M_\Phi (E_{ij}) \geq 0$ requires $M_\Phi (E_{ij})$ to be Hermitian. This is equivalent to $\lambda \in \mathbb R$, and the equation translates to:
        \begin{equation*}
            M_\Phi (E_{ij}) = \frac{1 - \lambda}{n} \mathbb I_n + \lambda (1 - \lambda) E_{jj} \geq 0 \, .
        \end{equation*}
        
        Then, $M_\Phi (E_{ij}) \geq 0$ if and only if the eigenvalues are positive, which reads $\frac{1 - \lambda}{n} \geq 0$, and $\frac{1 - \lambda}{n} + \lambda (1 - \lambda) \geq 0$, or equivalently $\lambda \in \left[-\frac{1}{n}, 1\right]$.

        \item If $i = j$: $\Phi(X) = \Phi(X^\dagger X) = \frac{1 - \lambda}{n} \mathbb I_n + \lambda E_{jj}$, and $\Phi(X)^\dagger = \frac{1 - \bar \lambda}{n} \mathbb I_n + \bar \lambda E_{jj}$. Then,
        \begin{equation*}
            M_\Phi (E_{ij}) = \frac{1 - \lambda}{n} \left(1 - \frac{1 - \bar \lambda}{n} \right) \mathbb I_n + \left[ \lambda - \frac{\lambda (1 - \bar \lambda) + \bar \lambda (1 - \lambda)}{n} - |\lambda|^2 \right] E_{jj} \geq 0 \, .
        \end{equation*}

        Again, $M_\Phi (E_{ij}) \geq 0$ requires $M_\Phi (E_{ij}) = M_\Phi (E_{ij})^\dagger$, which implies $\lambda \in \mathbb R$. Simplifying the equation and evaluating the positivity, we get        
        \begin{equation*}
            M_\Phi (E_{ij}) =\frac{1 - \lambda}{n} \left(1 - \frac{1 - \lambda}{n} \right) \mathbb I_n - \lambda (1 - \lambda) \left(1 - \frac{2}{n} \right) E_{jj} \geq 0 \, .
        \end{equation*}
        
        The matrix $M_\Phi (E_{ij}) $ is again diagonal, leading to conditions $\frac{1 - \lambda}{n} \left(1 - \frac{1 - \lambda}{n}\right) \geq 0$ and $\frac{1 - \lambda}{n} \left(1 - \frac{1 - \lambda}{n}\right) + \lambda \frac{1 - \lambda}{n} \left(1 - \frac{2}{n}\right) \geq 0$. This implies $\lambda \in \left[-\frac{1}{n - 1}, 1\right]$.
        
        The intersection between both intervals shows that if $\Phi$ is Schwarz, then $\lambda \in \left[-\frac{1}{n}, 1\right]$.
    \end{itemize}

    To obtain sufficient conditions,  we express $\Phi(X^\dagger X)$ and $\Phi(X)^\dagger \Phi(X)$ as a function of $X$ and $X_\perp \equiv X - \frac{\tr \, X}{n} \mathbb I_n$. It requires some algebraic manipulations to get    
    \begin{equation*}
        M_\Phi (X) = (1 - \lambda) \left[ \left(\frac{||X||_F^2}{n} - \frac{ | \tr \, X |^2}{n^2}\right) \mathbb I_n + \lambda X_\perp^\dagger X_\perp \right] \, ,
    \end{equation*}
    
   where $||X||_F^2  = \langle X, X \rangle = \tr \, X^\dagger X$ is the Frobenius norm, and $||X|| =  \sigma_1(X)$, the largest singular value of $X$, is the spectral norm.
   
   Since $1 - \lambda \geq 0$, then $M_\Phi(X)$ is positive if and only if the matrix $\left(\frac{||X||_F^2}{n} - \frac{|\tr \, X|^2}{n^2} \right) \mathbb I_n + \lambda X_\perp^\dagger X_\perp$ is positive.
   
   Then, we define the diagonalization of the non-negative Hermitian matrix $X_\perp^\dagger X_\perp$ to be $V_\perp D^2 V_\perp$ for some unitary $V_\perp$, and $D = \operatorname{diag} \{\sigma_1, \ldots, \sigma_n\}$. Note that if $A \geq 0$, then $VAV^\dagger \geq 0$ for all $V \in U(n)$. Therefore, conjugating by $V_\perp^\dagger$, the matrix to diagonalize simplifies:   
    \begin{equation*}
        \left(\frac{||X||_F^2}{n} - \frac{|\tr \, X|^2}{n^2}\right) \mathbb I_n + \lambda D^2 \, .
    \end{equation*}
   
  This matrix is diagonal, with $i$-th eigenvalue $\frac{||X||_F^2}{n} - \frac{|\tr \, X|^2}{n} + \lambda \sigma_i^2$. Finally, imposing the condition $\lambda \geq -\frac{1}{n}$:
  \begin{equation*}
      \frac{||X||_F^2}{n} - \frac{|\tr \, X|^2}{n} + \lambda \sigma_i^2 \geq
        \frac{||X||_F^2}{n} - \frac{|\tr \, X|^2}{n^2} - \frac{\sigma_1^2}{n} \geq \frac{||X||^2 - \sigma_1^2}{n} \geq 0\, ,
  \end{equation*}
where we have used the norm inequalities: $||X|| \leq ||X||_F \leq \sqrt{n} ||X||$, and the orthogonality relation $||X||_F^2 = ||X_\perp||_F^2 + \frac{|\tr \, X|^2}{n^2}$; for all $X \in M_n (\mathbb C)$.
\end{proof}

\begin{theorem} [Completely positive $U(n)$-equivariant maps]
\label{Theorem: Completely positive U-equivariant maps}
    Let $\Phi$ be a $U(n)$-equivariant unital map. Then $\Phi$ is completely positive if and only if $\lambda \in \left[-\frac{1}{n^2 - 1}, 1\right]$.
\end{theorem}
\begin{proof}
    To prove it we use the Choi matrix $C_\Phi$ associated to $\Phi$, that is $C_\Phi = \sum_{i,j} E_{ij} \otimes \Phi (E_{ij})$.
    
    First, notice that according to theorem \ref{Proposition: U-equivariant maps}, if $i \neq j$, then $\Phi(E_{ij}) = \lambda E_{ij}$. While if $i = j$, then $\Phi(E_{jj}) = \frac{1 - \lambda}{n} \mathbb I_n + \lambda E_{jj}$. With that, we can explicitly compute the Choi matrix    
    \begin{eqnarray}
    C_\Phi &=& \sum_i E_{ii} \otimes \left(\frac{1 + \lambda}{n} \mathbb I_n + \lambda E_{ii}\right) + \sum_{\substack{i, j \\ i \neq j}} E_{ij} \otimes \lambda E_{ij} \nonumber \\
    &=& \lambda \sum_{i, j} E_{ij} \otimes E_{ij} + \frac{1 - \lambda}{n} \mathbb I_{n^2} \, . \label{Equation: Choi U-equivariant}
    \end{eqnarray}
    
    The matrix is not diagonal, nevertheless, if $\nu$ are the eigenvalues of the first term $A := \sum_{i, j} E_{ij} \otimes E_{ij} $, then the eigenvalues of $C_\Phi$ are given by $\lambda \nu + \frac{1 - \lambda}{n}$. 
    
    To compute the eigenvalues of $A$, it is useful to see how it acts on the canonical basis $\{e_i \otimes e_j\}_{i, j = 1}^n$ of $M_n (\mathbb C) \otimes M_n (\mathbb C)$. In is not difficult to prove that $A e_k \otimes e_l = \delta_{kl} \sum_i e_i \otimes e_i$. We define    
    \begin{equation}
        v = \sum_k e_k \otimes e_k \, ,
    \end{equation}
    
    and we check:    
    \begin{equation}\label{Equation: Av}
        Av = \sum_{i, j, k} \delta_{jk} e_i \otimes e_i = n v \, .
    \end{equation}
    
    Therefore, the eigenvalues of $A$ are $n$ with multiplicity $1$; and $0$ with multiplicity $n^2 - 1$. To that end, note that    
    \begin{equation}\label{Equation: Aeiej}
        A (e_i \otimes e_j) = 0 \, , \qquad \forall \, i\neq j \, ,
    \end{equation} 
    
    and    
    \begin{equation}\label{Equation: A0}
        A w = 0 \, , \quad  \forall \,  w \in \mathrm{span\,} \{ e_i \otimes e_i \, \mid \, i = 1, \ldots, n \} \, , \ \text{such that } \, \langle w, v \rangle = 0 \, .
    \end{equation}
    
    Thus $0$ is an eigenvalue of $C_\Phi$ with multiplicity $(n^2 - n) + (n -1) = n^2 -1$. 
    
    Positivity of $C_\Phi$ thus translates to $\frac{1 - \lambda}{n} \geq 0$ and $n \lambda + \frac{1 - \lambda}{n} \geq 0$. Then $\Phi$ is CP if and only if $\lambda \in \left[-\frac{1}{n^2 - 1}, 1\right]$.
\end{proof}

After this result, we have identified a region of $U(n)$-equivariant maps that are S but not CP. These lie in the range $\lambda \in \left[-\frac{1}{n },-\frac{1}{n^2 - 1}\right)$, for all dimensions $n \geq 2$. This family proves itself interesting in the PPT$^{2}$ conjecture and witness entanglement (see Fig. \ref{fig:Un} below).

\begin{proposition}[PPT $\iff$ EB for $U(n)$-equivariant channels]
\label{Proposition: PPT iff EB for $U(n)$-equivariant channels}
Let $\,\Phi\colon M_n(\mathbb C)\to M_n(\mathbb C)$ be $U(n)$-equivariant and unital. Then $\Phi$ is PPT if and only if it is entanglement breaking (EB). Moreover, $\Phi$ is PPT and EB if and only if $\lambda \in \left[ -\frac{1}{n-1}, \, \frac{1}{n+1}\right]$.
\end{proposition}

\begin{proof}
We start by proving the second statement. Because of Eq. (\ref{Equation: Choi U-equivariant}) and (\ref{Equation: Av})-(\ref{Equation: A0}), the Choi matrix $C_\Phi$, is the isotropic state
\begin{equation*}
    C_\Phi = \frac{1 - \lambda}{n}\,\mathbb I_{n^2} + \lambda \, v^\dagger v \, .
\end{equation*}

Partial transpose gives $(\mathrm{id} \otimes T) (C_\Phi) = \frac{1 - \lambda}{n} \, \mathbb I_n + \lambda F$, where $F$ is the swap map. Since $F$ has eigenvalues $+1$ on the symmetric and $-1$ on the antisymmetric subspace, PPT is equivalent to $\frac{1 - \lambda}{n} \pm \lambda \ge 0$, i.e. $\lambda \in \left[- \frac{1}{n-1}, \frac{1}{n+1}\right]$.

For the first statement, separability of isotropic states (hence EB of $\Phi$) holds if and only if the fidelity satisfies $\mathcal F = v^\dagger \rho v \le 1$. Here $\mathcal F = (1 - \lambda) / n^2 + \lambda$, which yields the same interval. Thus PPT $\iff$ EB for $\Phi$ and $\lambda \in \left[-\frac{1}{n-1},\frac{1}{n+1}\right]$.
\end{proof}

\begin{figure}
  \centering
    \includegraphics[scale=0.5]{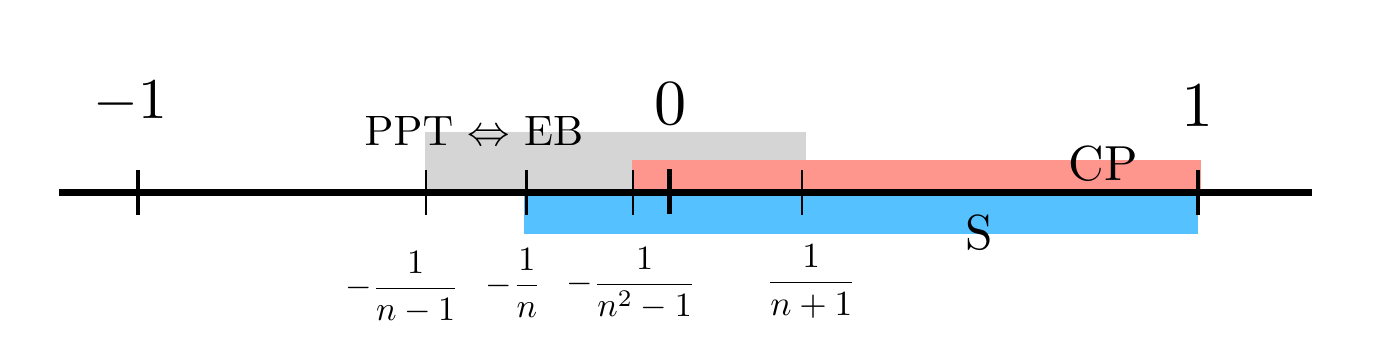}
  \caption{\small{A diagram representing the regions of (unital) $U(n)$-equivariant CP maps (red), (unital) $U(n)$-equivariant Schwarz maps (blue), and  $U(n)$-equivariant PPT maps (grey) which are equivalent to EB maps.} \label{fig:Un}}
\end{figure} 


\subsection{Diagonal unitary equivariance}
\label{Subsection: Diagonal unitary equivariance}

Instead of unitary equivariance, which is too restrictive, one can relax it to consider just diagonal unitary equivariance. The diagonal unitary group of degree $n$, is defined as $DU(n) \cong U(1) \times \cdots \times U(1) \subset U(n)$ and is the maximal diagonal torus of $U(n)$. Their elements are parametrized as $U(\varphi_1, \ldots, \varphi_n) = \operatorname{diag} \{e^{i\varphi_1}, \ldots, e^{i\varphi_n}\}$, with $\varphi_j \in [0, 2\pi)$, $j = 1, \ldots, n$. We also denote $u_j := e^{i\varphi_j}$. The group $DU(n)$ is Abelian and all its irreducible representations are one-dimensional. The irreducible representations of $DU(n)$ are defined by its characters
\begin{eqnarray*}
    \chi_\mathbf{k} \colon DU(n) &\to& U(1) \\
    U(\varphi_1, \ldots, \varphi_n) &\to& \chi_\mathbf{k} [U(\varphi_1, \ldots, \varphi_n)] = e^{i (k_1 \varphi_1 + \cdots + k_n \varphi_n )} = \prod_{j=1}^n u_j^{k_j}\, ,
\end{eqnarray*}
where $\mathbf{k} = (k_1 , \ldots, k_n) \in \mathbb{Z}^n$ and $\widehat{DU(n)} \cong \mathbb{Z}^n$. The group $DU(n)$ acts on $M_n(\mathbb{C})$ as a subgroup of $U(n)$, that is $\mu_U (X) = U X U^\dagger$, $U \equiv U(\varphi_1, \ldots, \varphi_n) \in DU(n)$.   

\begin{theorem} [$DU(n)$-equivariant maps]
\label{Theorem: DU-equivariant maps}
    Let $\Phi$ be a linear $DU(n)$-equivariant map. Then,
    \begin{equation*}
        [\Phi(X)]_{ij} = \begin{cases} \lambda_{ij} X_{ij} & \text{if} \ i \neq j,\\ \sum_{k} c_{ki} X_{kk} & \text{if} \ i = j.\end{cases}
    \end{equation*}
\end{theorem}

\begin{proof}
   We easily see that $[\mu_U(X)]_{ij} = u_i \bar{u}_j X_{ij}$ for all $X \in M_n (\mathbb C)$ and $U \equiv U(\varphi_1, \ldots, \varphi_n)$. Then, every one-dimensional subspace $W_{ij} = \operatorname{span}\{E_{ij}\}$ is $DU(n)$-invariant: $\mu_U(E_{ij}) = u_i \bar{u}_j E_{ij}$. Moreover $\operatorname{dim} W_{jk} = 1$, thus the representations $\mu^{(ij)}_U = u_i \bar{u}_j$ are all irreducible.
   
   \begin{itemize}
       \item If $i \neq j$, the irreducible representation $\mu^{(ij)}$ is the character $\chi_{\mathbf{k}_{ij}}$ with vector
       \begin{equation*}
           \mathbf{k}_{ij} = (0, \ldots, 0, \overset{(i)}{1}, 0, \ldots, 0, \overset{(j)}{-1}, 0, \ldots, 0) \, .
       \end{equation*}
       
       Then two irreducible representations $\mu^{(ij)}$ and $\mu^{(kl)}$, with $i \neq j$ and $k \neq l$, are equivalent if and only if $i = k$ and $j = l$.
       
       \item Moreover, if $i = j$, then  we get the irreducible representation $\mu^{(ii)} = 1$, with character is given by the vector $\mathbf{k}_{ii} = (0,\ldots, 0)$, and all representations $\mu^{(ii)}$, $i = 1, \ldots, n$, are equivalent.
   \end{itemize}
   
   The total space $M_n (\mathbb{C})$ is decomposed as a sum of (one-dimensional) irreducible subspaces as:
    \begin{equation}
        M_n (\mathbb C) \cong \bigoplus_{ij} W_{ij}, \qquad \mu \cong \bigoplus_{ij} \mu^{(ij)} \, ;
    \end{equation}  
    and because of theorem \ref{Theorem: Channel decomposition}, any linear $DU(n)$-equivariant map $\Phi$ admits a decomposition:
    \begin{equation}\label{eq:decompositionDU}
    \Phi = \sum_{i\neq j} \lambda_{ij} \, \mathrm{id}_{W_{ij}}  + \mathcal{C} \otimes \mathrm{id}_{\mathbb{C}} \, ,
    \end{equation}    
    where $\lambda_{ij} \in \mathbb C$ for all $i, j$; and $\mathcal{C} \colon W_0 \to W_0$ is a linear map, where $W_0$ is the $n$-dimensional subspace of $M_n(\mathbb{C})$ generated by the matrices $E_{ii}$, $i = 1, \ldots, n$.   Thus, $\mathcal C (E_{ii}) = \sum_{k= 1}^n  c_{ki} E_{kk}$.
    
    Then, applying $\Phi$, given by (\ref{eq:decompositionDU}), to $X$ gives the desired formulas. 
\end{proof}

In contrast to $U(n)$-equivariant maps, $DU(n)$-equivariant maps are parametrized by $\frac{n(n - 1)}{2} + n^2 = \frac{n(3n - 1)}{2}$ complex parameters.

\begin{lemma}[Unital $DU(n)$-equivariant channels]
\label{Theorem: Unital DU-eqivariant channels}
    Let $\Phi$ be $DU(n)$-equivariant. Then $\Phi$ is unital if and only if $\sum_j c_{ij} = 1$ for all $i = 1, \ldots, n$.
\end{lemma}
\begin{proof}
  Note that   $\mathbb I_n = \sum_i E_{ii} $, thus $\Phi(\mathbb I)  = \mathbb I$, is equivalent to $\sum_j c_{ij} = 1$, for all $i$.
\end{proof}

This means that if $\Phi$ is a $DU(n)$-equivariant map, then, the matrix associated to the operator $\mathcal C$ defined by its canonical decomposition (\ref{eq:decompositionDU}) is a stochastic matrix.

\begin{lemma} [Hermiticity preserving $DU(n)$-equivariant maps]
\label{Theorem: Hermiticity preserving DU-equivariant maps}
    Let $\Phi$ be $DU(n)$-equivariant. Then $\Phi$ is Hermiticity preserving if and only if $\lambda_{ij} = \bar{\lambda}_{ji}$ and $c_{ij} \in \mathbb R$ for all $i, j = 1, \ldots, n$.
\end{lemma}

\begin{proof}
    $\Phi$ is Hermiticity preserving if $[\Phi(X^\dagger)]_{ij} = [\Phi(X)^\dagger]_{ij}$ for all $X \in M_n (\mathbb C)$ and all $i, j = 1, \ldots, n$.

    \begin{itemize}
        \item If $i \neq j$, then $\lambda_{ij} \bar X_{ji} = \bar \lambda_{ji} \bar X_{ji}, \ \forall \, i,j \iff \lambda_{ij} = \bar{\lambda}_{ji}, \ \forall \, i,j$.
        \item If $i = j$, then $\sum_j c_{ij} \bar{X}_{jj} = \sum_j \bar{c}_{ij} \bar{X}_{jj}, \ \forall \, i,j$. Since this is true for all $X$, then it is true if and only if $c_{ij} = \bar{c}_{ij}$, $c_{ij} \in \mathbb R, \ \forall \, i, j$.  
    \end{itemize}
\end{proof}

\begin{theorem} [Schwarz $DU(n)$-equivariant maps]
\label{Theorem: Schwarz DU-equivariant maps}
    Let $\Phi$ be a $DU(n)$-equivariant unital and Hermicity preserving map. If $\Phi$ is Schwarz, then for all $i, j$ it is satisfied:
$$
 c_{kj} \in [0, 1]\, , \quad   k \neq j \, , \qquad 
 c_{jj} \geq |\lambda_{ij}|^2 \, .
$$
\end{theorem}
\begin{proof}
    Set $X = E_{ij}$.
    \begin{itemize}
        \item If $i \neq j$, $\Phi(X^\dagger X) = \Phi(E_{jj}) = \sum_k c_{kj} E_{kk}$ and $\Phi(X)^\dagger \Phi(X) = |\lambda_{ij}|^2 E_{jj}$. Then $M_\Phi(X) = \sum_k (c_{kj} - \delta_{ij} |\lambda_{ij}|^2) E_{kk}$. This matrix is diagonal with eigenvalues $c_{kj}$ for $k \neq j$; and $c_{jj} - |\lambda_{ij}|^2$ for $k = j$. Thus, $M_\Phi(X) \geq 0$ if and only if $c_{kj} \geq 0$ for $k \neq j$; and $c_{jj} \geq |\lambda_{ij}|^2 \geq 0$ for $k = j$.
        \item If $i = j$, then $\Phi(X)^\dagger \Phi(X) = \sum_k c_{kj}^2 E_{kk}$, thus $M_\Phi(X) = \sum_{k} c_{ki} (1 - c_{ki}) E_{kk}$. This matrix is diagonal with eigenvalues $c_{ki} (1 - c_{ki})$. Since $c_{ki} \geq 0$, then $M_\Phi(X) \geq 0$ if and only if $1 - c_{ki} \geq 0$.
    \end{itemize}
\end{proof}

It can be proven that more conditions on the coefficients exists. For instance, following the same technique used in theorem \ref{Theorem: Schwarz DU-equivariant maps} with $X = E_{ij} + E_{ik}$, one can get the condition:
\begin{equation}\label{eq:mixed}
    \left(c_{jj} + c_{jk} - |\lambda_{ij}|^2\right) \left(c_{jk} + c_{kk} - |\lambda_{ik}|^2\right) \geq |\lambda_{jk} + \lambda_{ji} \lambda_{ik}|^2 \, .
\end{equation}

Instead of following this approach and listing all conditions derived in this form, we will concentrate on particular families of $DU(n)$ equivariant maps and obtain explicit sufficient Schwarz conditions. But before doing that we will provide necessary and sufficient conditions for a $DU(n)$-equivariant map to be $CP$.

\begin{theorem} [Completely positive $DU(n)$-equivariant maps]
\label{Theorem: Completely positive DU-equivariant maps}
    Let $\Phi$ be a $DU(n)$-equivariant unital map. Then $\Phi$ is completely positive if and only if $c_{ij} \geq 0$ for all $i \neq j$ and $\tilde{C}_\Phi \geq 0$ where $\tilde C_\Phi \in M_n (\mathbb C)$ takes the form
    \begin{equation*}
        \tilde{C}_\Phi = \sum_{i = 1}^n c_{ii} E_{ii} + \sum_{\substack{i, j = 1 \\ i \neq j}}^n \lambda_{ij} E_{ij} \geq 0 \, .
    \end{equation*}
\end{theorem}
\begin{proof}
    To prove it we use the Choi matrix. If $i \neq j$, $\Phi(E_{ij}) = \lambda_{ij} E_{ij}$, while if $i = j$, $\Phi(E_{ii}) = \sum_j c_{ij} E_{j}$. Thus, allows to write the Choi matrix as
    \begin{gather*}
        C_\Phi = \sum_i E_{ii} \otimes \left(\sum_j c_{ij} E_{jj}\right) + \sum_{\substack{i, j \\ i \neq j}} E_{ij} \otimes \lambda_{ij} E_{ij} \\
        = \sum_{i, j} c_{ij} E_{ii} \otimes E_{jj} + \sum_{\substack{i, j \\ i \neq j}} \lambda_{ij} E_{ij} \otimes E_{ij} \, .
    \end{gather*}
    
    Then we have the first term, which is diagonal, and the second one which has scattered terms with no diagonal elements. If the map is Hermiticity preserving, then $C_\Phi$ would be Hermitian.
    
    This matrix can be diagonalized by blocks, leaving $n^2 - n$ eigenvalues equal to $c_{ij}$ with $i \neq j$, plus the $n$ eigenvalues of the $n \times n$ block
    \begin{equation*}
        \sum_{i = 1}^n c_{ii} E_{ii} + \sum_{\substack{i, j = 1 \\ i \neq j}}^n \lambda_{ij} E_{ij} := \tilde{C}_\Phi \, .
    \end{equation*}
\end{proof}

Particular cases of the eigenvalues could be found by fixing the values $c_{ii}$ or $\lambda_{ij}$.

\subsubsection{Qubits: Recovering the $DU(2)$ classification}

We now specialise the general form of a $DU(n)$-equivariant map to $n = 2$. Apply theorem \ref{Theorem: DU-equivariant maps} to the particular case, and obtain
\begin{equation}\label{eq:PhiX2}
\Phi (X) = \left(\begin{array}{cc} c_{11} X_{11} + c_{12} X_{22} & \lambda X_{12} \\ \bar{\lambda} X_{21} & c_{21} X_{11} + c_{22} X_{22}   \end{array} \right) \, ,
\end{equation}
where the map being unital and Hermiticity preserving translates to the $2\times 2$ matrix $\Gamma = (c_{ij})$, $i,j \in \{1, 2\}$ to be a real and stochastic: $\sum_i c_{ij} = 1$ and $\lambda_{12} = \lambda = \bar \lambda_{21}$.
This coincides exactly with Chru\'sci\'nski's form for diagonal-unitary (and orthogonal) covariance in $M_2(\mathbb{C})$ \cite{Chruscinski_24_1, Ch24}.

\begin{proposition}
Let $n = 2$, and thus $\Phi$ be the $DU(2)$-equivariant, unital, Hermiticity-preserving map given above by Eq. (\ref{eq:PhiX2}). Then:

\begin{itemize}
\item $\Phi$ is Schwarz if and only if 
\begin{equation}\label{eq:DU2S}
0 \leq c_{12}, c_{21} \leq 1 \, , \quad c_{11}, c_{22} \geq |\lambda|^2 \, ;
\end{equation}

where $c_{11} = 1 - c_{21}$, and $c_{22} = 1 - c_{12}$.

\item $\Phi$ is CP if and only if,
\begin{equation}\label{eq:DU2CP}
c_{12}, c_{21} \geq 0 \, , \quad c_{11} c_{22} \geq |\lambda|^2 \, .
\end{equation}
\end{itemize}
\end{proposition}

\begin{proof}
\begin{itemize}
    \item For Schwarz, application of theorem \ref{Theorem: Schwarz DU-equivariant maps}, necessary conditions given by the four scalar inequalities (\ref{eq:DU2S}) above are obtained.
    \item For complete positivity, theorem \ref{Theorem: Completely positive DU-equivariant maps} applies, leaving conditions (\ref{eq:DU2CP}). More precisely, computing the Choi matrix one gets
    \begin{equation*}
        C_\Phi =  \left(\begin{array}{cc} c_{11} & \lambda  \\ \bar{\lambda}  & c_{22}   \end{array} \right) \oplus \mathrm{diag\,}(c_{12}, c_{21}) \, , 
    \end{equation*}
    
    whose positivity gives the stated conditions (\ref{eq:DU2CP}).
\end{itemize}
\end{proof}

Hence the proposed general framework reproduces exactly the qubit classification in \cite{Chruscinski_24_1, Ch24}. The Schwarz but non-CP region corresponds to the open strip $|\lambda|^2 < c_{11}, c_{12}$, but $c_{11}c_{22} < |\lambda|^2 $.


\subsubsection{Qutrits: a symmetric $DU(3)$ family}

For $n =3$, the general parameter space is large (9 real parameters $c_{ij}$ and 3 complex parameters $\lambda_{ij}$), so we restrict to the permutation-symmetric subclass:
\begin{equation}
    \lambda_{ij} = \lambda \ (i \neq j) \, , \quad \Gamma = p \mathbb I_3 + \frac{1-p}{2} (1_3 - \mathbb I_3) \, ,
\end{equation}

where $1_3$ is the $3 \times 3$-matrix whose entries are all ones. 
Each column of $\Gamma$ has a diagonal entry $p$ and equal off-diagonal entries $(1-p)/2$. Thus, the parameters $p \in [0,1]$ and $\lambda \in \mathbb C$ completely determine the map. Imposing these conditions on Thm. \ref{Theorem: DU-equivariant maps}, we get
\begin{equation}\label{eq:PhiDU3}
        [\Phi(X)]_{ij} = \begin{cases} \lambda X_{ij} & \text{if} \ i \neq j,\\ p X_{ii} + \frac{1-p}{2} (X_{jj} + X_{kk}) & \{ i,j,k\} = \{ 1,2,3\} \, .\end{cases}
   \end{equation}
   
This class remains $DU(3)$-equivariant, unital and Hermiticity-preserving.

\begin{proposition}[Symmetric $DU(3)$ family: exact Schwarz and CP regions]\label{prop:DU3-symmetric-corrected}
Consider the permutation-symmetric $DU(3)$-equivariant, unital, Hermiticity-preserving map given by Eq. (\ref{eq:PhiDU3}).
Then:
\begin{enumerate}
\item $\Phi$ is Schwarz if and only if
\begin{equation}
    0 \le p \le 1, \quad |\lambda|^2 \le \frac{1}{2} \min \{p, \, 1 - p\} \, ,
\end{equation}

or equivalently, $p \ge 2 |\lambda|^2$ and $1 - p \ge 2 |\lambda|^2$.

\item $\Phi$ is completely positive if and only if
\begin{equation*}
    |\lambda| \le \min\{ p, \, 1-p\}
\end{equation*}

Hence the CP region is strictly contained in the Schwarz region.
\end{enumerate}
\end{proposition}

\begin{proof}
Fix a column $j$ and take $X = \alpha E_{ij} + \beta E_{kj}$ with $i \neq k \neq j$. Then $X^\dagger X=(|\alpha|^2+|\beta|^2)E_{jj}$ and $\Phi(X)=\lambda\alpha\,E_{ij}+\lambda\beta\,E_{kj}$. Thus the diagonal of $M_\Phi(X)=\Phi(X^\dagger X)-\Phi(X)^\dagger\Phi(X)$ is
\begin{eqnarray*}
(j, j)&:& p (|\alpha|^2 + |\beta|^2) - |\lambda|^2 |\alpha + \beta|^2 \, ,\\
(i, i)&:& \frac{1-p}{2} (|\alpha|^2 + |\beta|^2) - |\lambda|^2 |\alpha|^2 \, ,\\
(k, k)&:& \frac{1-p}{2} (|\alpha|^2 + |\beta|^2) - |\lambda|^2 |\beta|^2 \, .
\end{eqnarray*}

Positivity of $M_\Phi(X)$ for all $\alpha, \beta \in \mathbb C$ is equivalent to the non-negativity of these entries for all choices of phases.
Using $|\alpha + \beta|^2 \le 2 (|\alpha|^2 + |\beta|^2)$, we obtain the sharp bounds
\begin{equation*}
    p\ge 2|\lambda|^2 \, , \qquad \frac{1-p}{2}\ge|\lambda|^2 \, ,
\end{equation*}

as given in the theorem. They are sufficient (and necessary): equality is attained for $\alpha = \beta$ and for $\beta = 0$ (or $\alpha = 0$), respectively. By symmetry, the same conditions hold for row superpositions. In dimension $3$ these tests exhaust all off-diagonal supports (each column or
row contains exactly two off-diagonal elements), so the inequalities are both
necessary and sufficient for the Schwarz property.

The CP criterion follows from the positivity of the $2 \times 2$ blocks of the
Choi matrix (theorem \ref{Theorem: Completely positive DU-equivariant maps}), which yields $p, 1-p \ge |\lambda|$. Comparing the two regions gives $|\lambda| \le \min\{p, 1-p\}$ for CP and $|\lambda|^2 \le \tfrac{1}{2} \min\{p, 1-p\}$ for Schwarz. 
\end{proof}

Thus, a substantial portion of $DU(3)$-equivariant maps are Schwarz but not CP.

\begin{figure}
    \centering
    \includegraphics[scale=0.75]{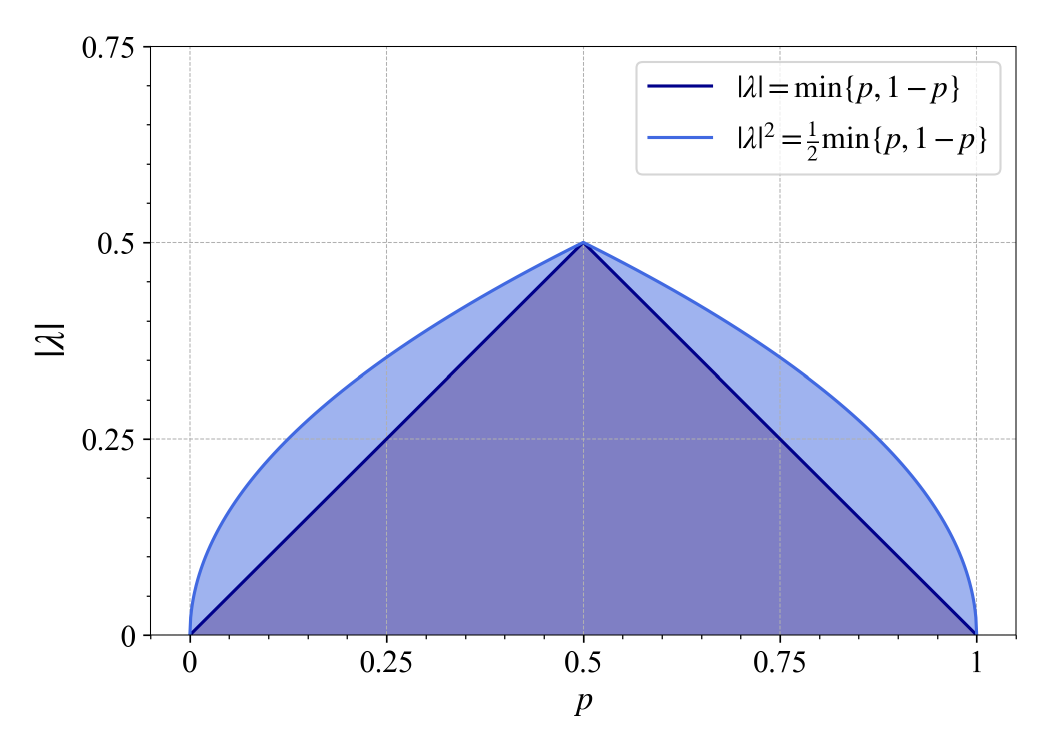}
    \caption{\small{Symmetric \(DU(3)\) family. Light: Schwarz region \( |\lambda| \le \sqrt{\tfrac12\,\min\{p,1-p\}} \). Dark: CP region \( |\lambda| \le \min\{p,1-p\} \).}}
    \label{fig:DU(3)}
\end{figure} 




\subsubsection{The PPT$^2$ conjecture for the $DU(2)$ and symmetric $DU(3)$ famileis}

Now we will apply the previous results to assess the PPT$^2$ conjecture directly in the case of $DU(2)$- and $DU(3)$-equivariant maps. Before that, we establish a general result concerning the ``squaring'' of the parameters of the maps.

\begin{lemma}[Composition law under $DU(n)$-equivariance]\label{lem:composition-laws}
Let $\Phi,\Psi: M_n(\mathbb C)\to M_n(\mathbb C)$ be $DU(n)$-equivariant, unital, Hermiticity-preserving, with
\begin{equation*}
[\Phi(X)]_{ij}=\begin{cases}\lambda^{\Phi}_{ij}X_{ij},& i\neq j,\\[1pt]\sum_k c^{\Phi}_{kj}\,X_{kk},& i=j,\end{cases}
\qquad
[\Psi(X)]_{ij}=\begin{cases}\lambda^{\Psi}_{ij}X_{ij},& i\neq j,\\[1pt]\sum_k c^{\Psi}_{kj}\,X_{kk},& i=j,\end{cases}
\end{equation*}

where $\lambda^{\Phi}_{ij}=\overline{\lambda^{\Phi}_{ji}}$, $\lambda^{\Psi}_{ij}=\overline{\lambda^{\Psi}_{ji}}$ and $\Gamma^\Phi=(c^{\Phi}_{kj})$, $\Gamma^\Psi=(c^{\Psi}_{kj})$ are column-stochastic.
Then $\Phi\circ\Psi$ is $DU(n)$-equivariant and
\begin{equation*}
\lambda^{\Phi\circ\Psi}_{ij}=\lambda^{\Phi}_{ij}\lambda^{\Psi}_{ij}\ (i\neq j)\,,\qquad
C^{\Phi\circ\Psi}=C^\Phi\,C^\Psi \,.
\end{equation*}
In particular, $\Phi^k$ has parameters $\big((\lambda_{ij})^k,\ (C)^k\big)$.
\end{lemma}

\begin{proof}
By direct application of Lemma \ref{Lemma: Composition of $G$-equivariant channels}. But in particular:
    \begin{itemize}
        \item For $i\neq j$, $\Psi(E_{ij})=\lambda^{\Psi}_{ij}E_{ij}$ and then $\Phi(\Psi(E_{ij}))=\lambda^{\Psi}_{ij}\lambda^{\Phi}_{ij}E_{ij}$.
        \item For $i=j$, $\Psi(E_{jj})=\sum_k c^{\Psi}_{kj}E_{kk}$ and $\Phi(\Psi(E_{jj}))=\sum_k c^{\Psi}_{kj}\sum_\ell c^{\Phi}_{\ell k}E_{\ell\ell}
=\sum_\ell (C^\Phi C^\Psi)_{\ell j}E_{\ell\ell}$.
    \end{itemize}
    Column-stochasticity is preserved by matrix multiplication, so unitality is preserved. Hermiticity preservation and $DU(n)$-equivariance are stable under composition.
\end{proof}

\begin{corollary}[Squaring law]\label{cor:DUsquaring}
For any $DU(n)$-equivariant unital, Hermiticity-preserving channel $\Phi$,
\[
\Phi^2:\quad \lambda\ \longmapsto\ \lambda^2\, ,
\]
entrywise on off-diagonals in $DU(n)$, and the classical mixing part evolves as $C\mapsto C^2$ in the $DU(n)$ case.
\end{corollary}

Now, we can prove that statement about the PPT$^2$ conjecture for diagonal-unitary equivariant maps. 

\begin{proposition}[PPT$^2$ in diagonal-unitary symmetry: $DU(2)$ and a $DU(3)$ slice]\label{prop:DU-PPT2}
Consider a $DU(n)$-equivariant, unital, Hermiticity-preserving channel with normal form given by Thm. (\ref{Theorem: DU-equivariant maps}). Then:
\begin{enumerate}
\item $DU(2)$ qubits. If $\Phi$ is PPT, then $\Phi^2$ is EB.
\item Symmetric $DU(3)$ family. On the two-parameter subclass
$\lambda_{ij}=\lambda$ $(i\ne j)$ and $C=pI+\frac{1-p}{2}(\mathbf 1-I)$ with $p\in[0,1]$,
if $\Phi$ is PPT, then $\Phi^2$ is EB.
\end{enumerate}
\end{proposition}

\begin{proof}
\emph{Block-PPT test.} For $DU(n)$, $(T\otimes\mathrm{id})(C_\Phi)$ decomposes into $2\times2$ blocks
$\begin{pmatrix}c_{ii}&\lambda_{ij}\\ \lambda_{ji}&c_{jj}\end{pmatrix}$ (plus diagonal pieces),
so PPT is implied by $c_{ij}\ge0$ and $c_{ii}c_{jj}\ge|\lambda_{ij}|^2$ for all $i\ne j$.

\emph{Composition rule.} The class is closed under composition and
$\Phi:\ (\lambda_{ij},C)\mapsto \Phi^2:\ (\lambda_{ij}^2,\ C^2)$.

\emph{EB criterion (diagonal dominance).} From Sect. 5.2 we have the sufficient EB test:
for each column $j$ (resp.\ row $i$), $c_{jj}\ge\sum_{i\ne j}|\lambda_{ij}|^2$ (resp.\ $c_{ii}\ge\sum_{j\ne i}|\lambda_{ij}|^2$).

\emph{(a) $DU(2)$.} Write $C=\begin{pmatrix}c_{11}&c_{12}\\ c_{21}&c_{22}\end{pmatrix}$ and $\lambda=\lambda_{12}$. PPT implies $c_{12},c_{21}\ge0$ and $c_{11}c_{22}\ge |\lambda|^2$.
After squaring: $\lambda\mapsto\lambda^2$, $C\mapsto C^2$ with $(C^2)_{11}\ge c_{11}^2$, $(C^2)_{22}\ge c_{22}^2$.
Thus $(C^2)_{11}\ge c_{11}^2\ge|\lambda|^2\ge|\lambda^2|^2$ and similarly for index 22, so the EB inequalities hold for $\Phi^2$.

\emph{(b) Symmetric $DU(3)$.} PPT reduces to $|\lambda|\le \min\{p,1-p\}$. After one squaring,
$|\lambda|\mapsto |\lambda|^2$ and $p\mapsto p':=(C^2)_{jj}=p^2+\tfrac12(1-p)^2$.
The EB test now requires $p'\ge 2|\lambda|^4$ (two off-diagonals per column), which holds because
$|\lambda|^4\le \min\{p^2,(1-p)^2\}$ and for $p\in[0,1]$ we have
$2\min\{p^2,(1-p)^2\}\le p^2+(1-p)^2 \le p^2+\tfrac12(1-p)^2=p'$ whenever $\min\{p,1-p\}\le \tfrac12$,
which is guaranteed by the PPT bound $|\lambda|\le \min\{p,1-p\}\le\tfrac12$.
Hence the EB inequalities hold for $\Phi^2$. 
\end{proof}


\subsection{Product unitary equivariance}
\label{Subsection: product unitary equivariance}

When imposing equivariance with respect to the unitary group, we got two complex parameters. When we relaxed the symmetry conditions, the number of parameters increased since more maps are available. Thus, in this last section, we chose an intermediate situation of considerable interest, the group $U(n_1) \times U(n_2)$ with $n_1 + n_2 = n$. Every $U \in U(n_1) \times U(n_2)$ takes the form $(U_1, U_2)$ where $U_{1, 2} \in U(n_{1, 2})$. This situation will correspond to a composite system where each one of its subsystems has a $U(n)$-symmetry.

Then it is easy to check that the representation $\mu_1 \otimes \mu_2$ is an irreducible representation of $G_1 \times G_2$ in the Hilbert space $\mathcal{H}_1 \otimes \mathcal{H}_2$ if and only if $\mu_1$ is an irreducible representation of $G_1$ on $\mathcal{H}_1$ and $\mu_2$ is an irreducible representation of $G_2$ on $\mathcal{H}_2$ \cite{Ibort_20}.   Then, using the decomposition (\ref{Equation: U decomposition}) of irreducible representations under the action of $U(n)$, the irreducible representations of the tensor product representation of $U(n_1) \times U(n_2)$ on $M_{n_1} (\mathbb C) \otimes M_{n_2} (\mathbb C)$ are given by:
\begin{eqnarray}
    M_{n_1} (\mathbb C) \otimes M_{n_2} (\mathbb C) &\cong&  [\mathbb C \mathbb I_{n_1} \oplus \mathfrak sl_{n_1}] \otimes [\mathbb C \mathbb I_{n_2} \oplus \mathfrak sl_{n_2}]  \nonumber \\  &\cong& [\mathbb C (\mathbb I_{n_1} \otimes \mathbb I_{n_2})] \oplus [\mathbb C \mathbb I_{n_1} \otimes \mathfrak sl_{n_2}]  \oplus\nonumber \\ &&  \oplus [\mathfrak sl_{n_1} \otimes \mathbb C \mathbb I_{n_2}] \oplus [\mathfrak sl_{n_1} \otimes \mathfrak sl_{n_2}]. \label{eq:tensorprod}
\end{eqnarray}

We denote the terms in the right hand side of  Eq. (\ref{eq:tensorprod}) by $W_{00}$, $W_{01}$, $W_{10}$, $W_{11}$, with dimensions $1$, $n_2^2 - 1$, $n_1^2 - 1$, and $(n_1^2 - 1)(n_2^2 - 1)$ respectively. They are all non-equivalent provided that $n_1 \neq n_2$. To write any $X \in M_{n_1} (\mathbb C) \otimes M_{n_2} (\mathbb C)$ in this decomposition, first decompose it as $\sum_i A_i \otimes B_i$. Now, write each term as a traceless and multiple of the identity part,
\begin{eqnarray*}
    A_i &=& \frac{\tr A_i}{n_1} \mathbb I_{n_1} + \left( A_i - \frac{\tr A_i}{n_1} \mathbb I_{n_1} \right) := \frac{\tr A_i}{n_1} \mathbb I_{n_1} + A_i^\perp \, , \\
    B_i &=& \frac{\tr B_i}{n_1} \mathbb I_{n_1} + \left( B_i - \frac{\tr B_i}{n_1} \mathbb I_{n_1} \right) := \frac{\tr B_i}{n_1} \mathbb I_{n_1} + B_i^\perp \, .
\end{eqnarray*}

Finally operate to get $X = X_{00} + X_{01} + X_{10} + X_{11}$ with:
\begin{eqnarray}
    X_{00} &=& \frac{\sum_i \tr \, A_i \tr \, B_i}{n_1 n_2} \mathbb I_{n_1} \otimes \mathbb I_{n_2} \in W_{00} \, . \label{eq:X00} \\
    X_{01} &=& \frac{1}{n_1} \mathbb I_{n_1} \otimes \sum_i \tr \, (A_i) B_i^\perp \in W_{01} \, .\label{eq:X01}  \\
    X_{10} &=& \sum_i \tr \, (B_i) A_i^\perp \otimes \frac{1}{n_2} \mathbb I_{n_2} \in W_{10} \, .  \label{eq:X10} \\
    X_{11} &=& \sum_i  A_i^\perp \otimes B_i^\perp  \in W_{11} \, . \label{eq:X11} 
\end{eqnarray}

\begin{theorem} [$U (n_1) \times U (n_2)$-equivariant maps]
\label{Theorem: U1xU2-equivariant maps}
    Let $\Phi$ be a linear $U (n_1) \times U (n_2)$-equivariant map. Then, $\Phi(X) = \sum_{a, b = 0}^1 \lambda_{ab} X_{ab}$, with $X_{ab}$ as in Eqs. (\ref{eq:X00}-\ref{eq:X11}).
\end{theorem}
\begin{proof}
    Apply Schur's lemma (lemma \ref{Lemma: Schur's Lemma}) to the decomposition given in (\ref{eq:tensorprod}).
\end{proof}

Thus, we have four parameters $\lambda_{00}, \lambda_{01}, \lambda_{10}, \lambda_{11}$, two more than in the unitary case.
\begin{lemma}[Unital $U (n_1) \times U (n_2)$-equivariant channels]
\label{Theorem: Unital U1xU2-eqivariant channels}
    Let $\Phi$ be $U (n_1) \times U (n_2)$-equivariant. Then it is unital if and only if $\lambda_{00} = 1$.
\end{lemma}
\begin{proof}
    \begin{equation*}
        \Phi(\mathbb I_{n_1} \otimes \mathbb I_{n_2}) = \lambda_{00} \, \mathbb I_{n_1} \otimes \mathbb I_{n_2} = \mathbb I_{n_1} \otimes \mathbb I_{n_2} \iff \lambda_{00} = 1 \, .
    \end{equation*}
\end{proof}

\begin{lemma} [Hermiticity preserving $U (n_1) \times U (n_2)$ - equivariant maps]
\label{Theorem: Hermiticity preserving U1xU2-equivariant maps}
    Let $\Phi$ be $U (n_1) \times U (n_2)$-equivariant. Then it is Hermiticity preserving if and only if $\lambda_{ij} \in \mathbb R$ for all $i, j = 0, 1$.
\end{lemma}
\begin{proof}
    \begin{gather*}
        \Phi\left(X^\dagger\right) - \Phi(X)^\dagger = \sum_{i, j = 0}^1 \left(\lambda_{ij} - \bar{\lambda}_{ij}\right) X_{ij}^\dagger = 0 \\
        \iff \lambda_{ij} = \bar{\lambda}_{ij} \, , \quad \forall \, i, j \iff \lambda_{ij} \in \mathbb R\, , \quad \forall \, i, j \, .
    \end{gather*}
    
    Where we have used the orthogonality of the subspaces.
\end{proof}

For the next two theorems, we will benefit from exploring how the map acts on the canonical basis. To do so, we denote:
\begin{eqnarray}
a &\equiv& \frac{1 - \lambda_{01} - \lambda_{10} + \lambda_{11}}{n_1 n_2} \, , \quad b \equiv \frac{\lambda_{01} - \lambda_{11}}{n_1} \, , \nonumber \\ c &\equiv& \frac{\lambda_{10} - \lambda_{11}}{n_2}\, , \quad  d \equiv \lambda_{11} \label{eq:abcd} \, .
\end{eqnarray}

Consider $\{E_{ij}\}$ and $\{F_{kl}\}$ to be the canonical basis of $M_{n_1} (\mathbb C)$ and $M_{n_2} (\mathbb C)$ respectively.  Then we get:

\begin{enumerate}
\item  $i \neq j$, $j \neq l$: $\Phi(E_{ij} \otimes F_{kl}) = d E_{ij} \otimes F_{kl}$.

\item $i = j$, $j \neq l$: $\Phi(E_{ii} \otimes F_{kl}) = b \mathbb I_{n_1} \otimes F_{kl} + d E_{ij} \otimes F_{kl}$.

\item $i \neq j$, $j = l$: $\Phi(E_{ij} \otimes F_{kk}) = c E_{ij} \otimes \mathbb I_{n_2} + d E_{ij} \otimes F_{kl}$.

\item $i = j$, $j = l$: $\Phi(E_{ij} \otimes F_{kk}) = a \mathbb I_{n_1 n_2} + b \mathbb I_{n_1} \otimes F_{kk} + c E_{ii} \otimes \mathbb I_{n_2} + d E_{ii} \otimes F_{kk}$.
\end{enumerate}

First, we are going to consider necessary conditions for Schwarz maps. To do so, we evaluate at $X = E_{ij} \otimes F_{kl}$.

\begin{enumerate}
    \item $i \neq j$, $j \neq l$: $M_\Phi = a \mathbb I_{n_1 n_2} + b \mathbb I_{n_1} \otimes F_{ll} + c E_{jj} \otimes \mathbb I_{n_2} + d E_{jj} \otimes F_{ll}$. This operator is diagonal. To compute its eigenvalues, we consider its action in the canonical basis. $M_\Phi e_p \otimes e_q = a e_p \otimes e_q + \delta_{ql} b e_p \otimes e_l + \delta_{pj} c e_j \otimes e_q + \delta_{pj} \delta_{ql} d (1 - d) e_j \otimes e_l$. Then, the conditions are: 
\begin{equation}\label{eq:conditions1}
a \geq 0\, ,  \quad a + c \geq 0 \, ,  \quad a + b \geq 0 \, , \quad a + b + c + d (1 - d) \geq 0\, .
\end{equation}

We follow a similar procedure with the remaining cases.

\item  $i = j$, $j \neq l$: $M_\Phi = a \mathbb I_{n_1 n_2} + b (1 - b) \mathbb I_{n_1} \otimes F_{ll} + c E_{jj} \otimes \mathbb I_{n_2} + d (1 - 2 b - d) E_{jj} \otimes F_{ll}$. The new non-trivial conditions are: 
\begin{equation}\label{eq:conditions2}
a + b (1 - b) \geq 0\, , \quad a + b (1 - b) + c + d (1 - 2b - d) \geq 0 \, .
\end{equation}

\item $i \neq j$, $j = l$: $M_\Phi = a \mathbb I_{n_1 n_2} + b \mathbb I_{n_1} \otimes F_{ll} + c (1 - c) E_{jj} \otimes \mathbb I_{n_2} + d (1 - 2 c - d) E_{jj} \otimes F_{ll}$. The new non-trivial conditions are: 
\begin{equation}\label{eq:conditions3}
a + c (1 - c) \geq 0 \, \quad  a + b + c (1 - c) + d (1 - 2c - d) \geq 0\, .
\end{equation}

\item $i = j$, $j = l$: $\Phi(X^\dagger X) = \Phi(X)$, then the eigenvalues of $M_\Phi = \Phi(X) (1 - \Phi(X))$ are $\nu (1 - \nu)$ where $\nu$ are the eigenvalues of $\Phi(X)$. The new conditions are: 
\begin{eqnarray}
a (1 - a) \geq 0\, ,&&   \quad (a + c) [1 - (a + c)] \geq 0 \, , \label{eq:conditions4}   \\ (a + b) [1 - (a + b)] \geq 0 \, , &&\quad (a + b + c + d) [1 - (a + b + c + d)] \geq 0\, .\nonumber
\end{eqnarray}
\end{enumerate}

\begin{theorem} [Schwarz $U (n_1) \times U (n_2)$-equivariant maps]
\label{Theorem: Schwarz U1xU2-equivariant maps}
    Let $\Phi$ be a $U (n_1) \times U (n_2)$-equivariant unital map. If $\Phi$ is Schwarz, then conditions given by (\ref{eq:conditions1}-\ref{eq:conditions4}) are satisfied.
\end{theorem}

Now we compute the Choi matrix. To do so, we decompose the sum $\sum_{i, j, k, l} E_{ij} \otimes F_{kl} \otimes \Phi(E_{ij} \otimes F_{kl})$ into four terms mediated by the same four conditions used above. After simplification, one gets:
\begin{eqnarray}
C_\Phi &=& a \mathbb I_{(n_1 n_2)^2} + b \sum_{k, l} (\mathbb I_{n_1} \otimes F_{kl})^{\otimes \ 2} + c \sum_{i, j} (E_{ij} \otimes \mathbb I_{n_2})^{\otimes 2} \nonumber \\ &+& d \sum_{i, j, k, l} (E_{ij} \otimes F_{kl})^{\otimes 2} \, .
\end{eqnarray}

We have no exact expression for its eigenvalues, but it is possible to compute their exact value for small dimension cases $(n_1, n_2) = (2, 2), (2, 3), (3, 2), (3, 3)$.

\begin{lemma}\label{lem:eigChoi23}
For $n_1 = 2,3$; $n_2 = 2,3$, the eigenvalues of the Choi matrix $C_\Phi$ of a unital, hermiticity-preserving $U(n_1) \otimes U(n_2)$-equivariant map are given by:
\begin{equation}\label{eq:eigChoi23a}
\frac{1}{n_1 n_2} [1 + (n_2^2 - 1) \lambda_{01} + (n_1^2 - 1) \lambda_{10} + (n_1^2 - 1)(n_2^2 - 1) \lambda_{11}] \, , 
\end{equation}
with multiplicity 1.
\begin{equation}\label{eq:eigChoi23b}
\frac{1}{n_1 n_2} [1 + (n_2^2 - 1) \lambda_{01} - \lambda_{10} - (n_2^2 - 1) \lambda_{11}]\, ,
\end{equation} 
with multiplicity $n_1^2 - 1$.
\begin{equation}\label{eq:eigChoi23c}
\frac{1}{n_1 n_2} [1 - \lambda_{01} + (n_1^2 - 1) \lambda_{10} - (n_1^2 - 1) \lambda_{11}]\, ,
\end{equation}
 with multiplicity $n_2^2 - 1$.
\begin{equation}\label{eq:eigChoi23d}
\frac{1}{n_1 n_2} [1 - \lambda_{01} - \lambda_{10} + \lambda_{11}] \, , 
\end{equation}
with multiplicity $(n_1^2 - 1)(n_2^2 - 1)$.
\end{lemma}

In the following paragraphs we summarise the main results for $n_1 = 2$, $n_2 = 2,3$.

\begin{proposition}[Tensor-unitary equivariance: $(n_1,n_2)=(2,2)$]\label{prop:U2U2}
Let $\Phi$ be a $U(2)\otimes U(2)$-equivariant, unital, Hermiticity-preserving map.
Then:
\begin{enumerate}
\item $\Phi$ is completely positive if and only if the following four linear inequalities hold:
\begin{eqnarray}
1+3\lambda_{01}+3\lambda_{10}+9\lambda_{11}\ & \ge & 0 \, ,\\
1+3\lambda_{01}-\lambda_{10}-3\lambda_{11}\ & \ge & 0 \, ,\\
1-\lambda_{01}+3\lambda_{10}-3\lambda_{11}\ & \ge & 0 \, ,\\
1-\lambda_{01}-\lambda_{10}+\lambda_{11}\ & \ge & 0 \, .
\end{eqnarray}

\item If $\Phi$ is Schwarz, then the following specialized constraints from (\ref{eq:conditions1}-\ref{eq:conditions4}), hold:
\[
\begin{aligned}
1-\lambda_{01}-\lambda_{10}+\lambda_{11}\ &\ge& 0 \, ,\\
1-\lambda_{01}+\lambda_{10}-\lambda_{11}\ &\ge& 0 \, ,\\
1+\lambda_{01}-\lambda_{10}-\lambda_{11}\ &\ge& 0 \, ,\\
1+\lambda_{01}+\lambda_{10}+\lambda_{11}-4\lambda_{11}^{2}\ &\ge& 0 \, ,
\end{aligned}
\qquad
\begin{aligned}
\lambda_{01}^{2}-2\lambda_{01}\lambda_{11}-\lambda_{01}+\lambda_{10}+\lambda_{11}^{2}+\lambda_{11}-1\ &\le& 0 \, ,\\
\lambda_{01}^{2}+2\lambda_{01}\lambda_{11}-\lambda_{01}-\lambda_{10}+\lambda_{11}^{2}-\lambda_{11}-1\ &\le& 0 \, ,\\
\lambda_{10}^{2}-2\lambda_{10}\lambda_{11}-\lambda_{10}+\lambda_{01}+\lambda_{11}^{2}+\lambda_{11}-1\ &\le& 0 \, ,\\
\lambda_{10}^{2}+2\lambda_{10}\lambda_{11}-\lambda_{10}-\lambda_{01}+\lambda_{11}^{2}-\lambda_{11}-1\ &\le& 0 \, ,
\end{aligned}
\]
\[
\begin{aligned}
\lambda_{01}+\lambda_{10}-\lambda_{11}\ &\in [-3,1]\, , & \quad
\lambda_{01}-\lambda_{10}+\lambda_{11}\ &\in [-3,1]\, , &\\
\lambda_{01}+\lambda_{10}+\lambda_{11}\ &\in [-1,3]\, , &\quad
-\lambda_{01}+\lambda_{10}+\lambda_{11}\ &\in  [-1,3] \, . &
\end{aligned}
\]
\end{enumerate}
\end{proposition}

\begin{proof}
The CP part follows by imposing nonnegativity of the four Choi eigenvalue families
given in Lemma \ref{lem:eigChoi23} with $(n_1^2-1,n_2^2-1,n_1n_2)=(3,3,4)$.  
The Schwarz implications are the $(2,2)$ specialization of the necessary inequalities
(\ref{eq:conditions1}-\ref{eq:conditions4}). 
\end{proof}

\begin{proposition}[Tensor-unitary equivariance: $(n_1,n_2)=(2,3)$]\label{prop:U2U3}
Let $\Phi$ be a $U(2)\otimes U(3)$-equivariant, unital, Hermiticity-preserving map.
Then:
\begin{enumerate}
\item $\Phi$ is completely positive if and only if the following four linear inequalities hold:
\[
\begin{aligned}
1+8\lambda_{01}+3\lambda_{10}+24\lambda_{11}\ &\ge 0 \, ,&\\
1+8\lambda_{01}-\lambda_{10}-8\lambda_{11}\ &\ge 0 \, ,&\\
1-\lambda_{01}+3\lambda_{10}-3\lambda_{11}\ &\ge 0 \, ,&\\
1-\lambda_{01}-\lambda_{10}+\lambda_{11}\ &\ge 0 \, .&
\end{aligned}
\]
\item If $\Phi$ is Schwarz, then the following specialized constraints from (\ref{eq:conditions1}-\ref{eq:conditions4}) hold:
\[
\begin{aligned}
1-\lambda_{01}-\lambda_{10}+\lambda_{11}\ &\ge 0 \, ,&\\
1-\lambda_{01}+\lambda_{10}-\lambda_{11}\ &\ge 0 \, ,&\\
1+2\lambda_{01}-\lambda_{10}-2\lambda_{11}\ &\ge 0 \, ,&\\
1+2\lambda_{01}+\lambda_{10}+2\lambda_{11}-6\lambda_{11}^{2}\ &\ge 0 \, ,&
\end{aligned}
\]
\[
\begin{aligned}
\lambda_{01}^{2}-2\lambda_{01}\lambda_{11}-\lambda_{01}+\lambda_{10}+\lambda_{11}^{2}+\lambda_{11}-1\ & \le 0 \, &\\
\lambda_{10}^{2}-\tfrac{2}{3}\lambda_{10}\lambda_{11}-\lambda_{10}+\tfrac{2}{3}\lambda_{01}
+\tfrac13\lambda_{11}^{2}+\tfrac23\lambda_{11}-\tfrac13\ & \le 0\, ,&
\end{aligned}
\]
together with their sign-reflected counterparts obtained as in (\ref{eq:conditions2}-\ref{eq:conditions3})  and the range constraints
\[
\begin{aligned}
\lambda_{01}+\lambda_{10}-\lambda_{11}\ & \in [-5,1] \, , &\quad
\lambda_{01}-\lambda_{10}+\lambda_{11}\ & \in [-5,1] \, , &\\
2\lambda_{01}+\lambda_{10}+2\lambda_{11}\ & \in [-1,5] \, , &\quad
-2\lambda_{01}+\lambda_{10}+2\lambda_{11}\ & \in [-1,5] \, . &
\end{aligned}
\]
\end{enumerate}
\end{proposition}

\begin{proof}
The CP part is obtained by substituting $(n_1^2-1,n_2^2-1,n_1n_2)=(3,8,6)$
into the four Choi eigenvalue families listed in Lemma \ref{lem:eigChoi23}.  
The Schwarz implications are the $(2,3)$ specialization of (\ref{eq:conditions1}-\ref{eq:conditions4}).
\end{proof}

\begin{corollary}[Geometry for $U(2)\otimes U(2)$]\label{cor:geom-22}
In the parameter space $(\lambda_{01},\lambda_{10},\lambda_{11})\in\mathbb{R}^3$ with $\lambda_{00}=1$:
\begin{enumerate}
\item \textbf{CP region.} The set of completely positive maps is the nonempty, closed, convex \emph{polyhedron}:
\[
\mathsf{CP}_{2,2}\ =\ \bigcap_{k=1}^4 \{\; \langle \alpha^{(k)},\lambda\rangle \ge -1 \;\},
\]
where $\lambda=(\lambda_{01},\lambda_{10},\lambda_{11})$ and the four normals are
\[
\alpha^{(1)}=(3,3,9),\quad
\alpha^{(2)}=(3,-1,-3),\quad
\alpha^{(3)}=(-1,3,-3),\quad
\alpha^{(4)}=(-1,-1,1),
\]
corresponding to the four linear inequalities in Proposition~\ref{prop:U2U2}. Hence $\mathsf{CP}_{2,2}$ is a convex polyhedral cone (after translating by the origin constraint), with four supporting planes. \emph{(All constraints arise from the four Choi eigenvalue families.)} 
\item \textbf{Schwarz envelope.} The Schwarz-feasible set $\mathsf{S}_{2,2}$ is a strictly larger closed, convex set containing $\mathsf{CP}_{2,2}$. Besides three of the CP supporting planes (the first three linear inequalities in Proposition~\ref{prop:U2U2}), one boundary component is the \emph{quadratic} surface
\[
1+\lambda_{01}+\lambda_{10}+\lambda_{11}-4\lambda_{11}^{2}=0,
\]
with inward normal varying along the surface. The additional quadratic and interval bounds in (S-quad)?(S-rng) carve the curved face and cap the set; in particular (S-quad) yields nonplanar supporting surfaces, showing $\mathsf{S}_{2,2}$ is not polyhedral. \emph{(All inequalities are the $(2,2)$ specialization of (\ref{eq:conditions1}-\ref{eq:conditions4}) )}
\end{enumerate}
\end{corollary}

\begin{corollary}[Geometry for $U(2)\otimes U(3)$]\label{cor:geom-23}
In the parameter space $(\lambda_{01},\lambda_{10},\lambda_{11})\in\mathbb{R}^3$ with $\lambda_{00}=1$:
\begin{enumerate}
\item \textbf{CP region.} The completely positive set is the closed convex \emph{polyhedron}
\[
\mathsf{CP}_{2,3}\ =\ \bigcap_{k=1}^4 \{\; \langle \beta^{(k)},\lambda\rangle \ge -1 \;\},
\]
with
\[
\beta^{(1)}=(8,3,24),\quad
\beta^{(2)}=(8,-1,-8),\quad
\beta^{(3)}=(-1,3,-3),\quad
\beta^{(4)}=(-1,-1,1),
\]
exactly the four half-spaces in Proposition~\ref{prop:U2U3}(CP iff). Thus $\mathsf{CP}_{2,3}$ is polyhedral with four facets determined by the Choi eigenvalue families. 
\item \textbf{Schwarz envelope.} The Schwarz-feasible set $\mathsf{S}_{2,3}$ strictly contains $\mathsf{CP}_{2,3}$ and is bounded by three linear faces (the first three linear inequalities in (S-lin\(_{2,3}\))) together with a \emph{quadratic} face
\[
1+2\lambda_{01}+\lambda_{10}+2\lambda_{11}-6\lambda_{11}^{2}=0,
\]
and further quadratic caps from (S-quad\(_{2,3}\)) and the interval constraints (S-rng\(_{2,3}\)). Hence $\mathsf{S}_{2,3}$ is non-polyhedral; its curved boundary stems from the nontrivial Schwarz conditions (\ref{eq:conditions1}-\ref{eq:conditions4})  once $n_1\neq n_2$. 
\end{enumerate}
\end{corollary}

\begin{remark}[Visuals and extremal structure]
For both $(2,2)$ and $(2,3)$, the CP region is the intersection of four half-spaces (a tetrahedral-like polyhedron in $(\lambda_{01},\lambda_{10},\lambda_{11})$-space after an affine rescaling).  
The Schwarz region is its convex enlargement with one prominent \emph{curved} face along a quadratic surface in $\lambda_{11}$, together with quadratic tightenings that bend the boundary away from the CP facets. In particular, any point on the curved face is \emph{not} exposed by a single linear Choi inequality, certifying that S $\supsetneq$ CP in these symmetry classes. \emph{(All faces and bounds follow by specializing (18)-(21) and the Choi spectra stated after Theorem~5.14.)} 
\end{remark}

\begin{remark}
Our results suggest, and thus we conjecture, that the eigenvalues of the Choi matrix $C_\Phi$ are given by Eqs. (\ref{eq:eigChoi23a})-(\ref{eq:eigChoi23d}).
\end{remark}

As in the case of diagonal-unitary equivariant maps we may inquire about the PPT$^2$ conjecture for product unitary symmetry in small dimensions.

\begin{lemma}[Composition law under  $U(n_1)\otimes U(n_2)$ equivariance]\label{lem:composition-laws2}
Let $\Phi,\Psi$ be $U(n_1)\otimes U(n_2)$-equivariant, unital, Hermiticity-preserving, acting diagonally on the isotypic decomposition:
\begin{equation*}
    M_{n_1}(\mathbb C) \otimes M_{n_2} (\mathbb C) = W_{00} \oplus W_{01} \oplus W_{10} \oplus W_{11}
\end{equation*}

with real parameters $\lambda^{\Phi}_{\alpha\beta}$, $\lambda^{\Psi}_{\alpha\beta}$ ($\alpha, \beta = 0, 1$) and $\lambda^{\Phi}_{00} = \lambda^{\Psi}_{00} = 1$. Writing $P_{\alpha\beta}$ for the equivariant projectors onto $W_{\alpha\beta}$, we have
\begin{equation*}
    \Phi = \sum_{\alpha, \beta} \lambda^{\Phi}_{\alpha\beta} \, P_{\alpha\beta} \, , \quad \Psi = \sum_{\alpha,\beta} \lambda^{\Psi}_{\alpha\beta} \, P_{\alpha\beta} \, .
\end{equation*}

Then $\Phi \circ \Psi = \sum_{\alpha,\beta} (\lambda^{\Phi}_{\alpha\beta} \lambda^{\Psi}_{\alpha\beta}) \, P_{\alpha\beta}$, i.e.
\begin{equation*}
\lambda^{\Phi\circ\Psi}_{\alpha\beta} = \lambda^{\Phi}_{\alpha\beta} \lambda^{\Psi}_{\alpha\beta} \quad (\alpha,\beta = 0,1 ) \, .
\end{equation*}

In particular, $\Phi^k$ acts with parameters $(\lambda_{\alpha\beta})^k$.
\end{lemma}

\begin{proof}
By direct application of Lemma \ref{Lemma: Composition of $G$-equivariant channels}. But in particular, equivariance implies Schur-type diagonal action on each isotypic component; the $P_{\alpha\beta}$ are orthogonal idempotents commuting with both maps. Thus $\Phi\circ\Psi$ acts by the product of the scalars on each block, yielding the stated formula and the power law.
\end{proof}

\begin{corollary}[Squaring law]\label{cor:prod_squaring}
For any $U(n_1)\otimes U(n_2)$-equivariant unital, Hermiticity preserving channel $\Phi$,
\begin{equation*}
    \Phi^2 \colon \lambda \longmapsto \lambda^2\, ,
\end{equation*}

per isotypic coefficient in $U(n_1)\otimes U(n_2)$.
\end{corollary}

\begin{theorem}[PPT$^2$ for product-unitary symmetry in small dimensions]\label{thm:U2U2-U2U3-PPT2}
Let $\Phi$ be a $U(n_1)\otimes U(n_2)$-equivariant, unital, Hermiticity-preserving channel on $M_{n_1} (\mathbb C) \otimes M_{n_2} (\mathbb C)$, acting diagonally on the isotypic sum $W_{00} \oplus W_{01} \oplus W_{10} \oplus W_{11}$ with parameters $\lambda_{00} = 1$ and $\lambda_{01}, \lambda_{10}, \lambda_{11}\in\mathbb R$. For $(n_1, n_2) = (2, 2), (2, 3)$, if $\Phi$ is PPT then $\Phi^2$ is EB.
\end{theorem}

\begin{proof}
(i) PPT planes. In both cases the Choi matrix block-diagonalizes according to the four isotypic components. Its partial transpose preserves the decomposition and flips the sign on the antisymmetric sector, hence PPT is equivalent to
the four linear half-space constraints listed in Propositions~\ref{prop:U2U2} and \ref{prop:U2U3} (the same planes used there for CP, now as PPT planes).

(ii) Composition. Equivariance implies $\Phi^2$ acts with squared parameters:
$(\lambda_{01}, \lambda_{10}, \lambda_{11}) \mapsto (\lambda_{01}^2, \lambda_{10}^2, \lambda_{11}^2)$.

(iii) EB after one squaring. The EB-sufficient Schwarz bounds specialized in
Propositions~\ref{prop:U2U2}-\ref{prop:U2U3} include a quadratic ``bowed'' face of the form $\alpha_0 + \alpha_1 \lambda_{01} + \alpha_2 \lambda_{10} + \alpha_3 \lambda_{11} - \beta \, \lambda_{11}^2 \ge 0$ (together with symmetric companions). Since the PPT region already enforces four linear bounds through the origin, replacing $\lambda$ by $\lambda^2$ moves every PPT point strictly toward the origin and across those quadratic EB-faces. Equivalently, plugging $(\lambda_{01}^2, \lambda_{10}^2, \lambda_{11}^2)$ into the boxed EB inequalities of Propositions~\ref{prop:U2U2}-\ref{prop:U2U3} makes each inequality strictly easier, hence they hold for $\Phi^2$.
Therefore $\Phi^2$ is EB. 
\end{proof}

\section{Conclusion}

We have determined the structure and parametrization of equivariant maps with respect to any compact group. Using it, we developed a detailed classification of unital, Hermiticity-preserving Schwarz maps constrained by symmetry. For the full unitary group $U(n)$, the classification is complete and shows that these channels form a one-parameter family interpolating between the identity and the completely depolarizing channel. In this class, $\mathrm{PPT}\!\iff\!\mathrm{EB}$, giving a direct and transparent proof of the $\mathrm{PPT}^2$ conjecture.

For the diagonal unitary subgroup $DU(n)$ we derived a hierarchy of algebraic conditions based on principal minors that yield both necessary and sufficient Schwarz criteria in low dimension. In particular, the complete description of $DU(2)$ and a symmetric $DU(3)$ family shows how PPT maps become entanglement-breaking upon squaring. The tensor-product symmetries $U(2)\!\otimes\!U(2)$ and $U(2)\!\otimes\!U(3)$ exhibit the same behaviour: composition squares the symmetry parameters, forcing the maps into the entanglement-breaking region.

Altogether, the results prove that the general structure theorem for symmetric channels can be applied to determine the description of information channels and to simplify the resolution of some problems. Indeed, we show that they provide explicit, symmetry-based mechanisms explaining why certain structured channels obey $\mathrm{PPT}^2$. They also supply exact examples where the boundaries between positivity, Schwarz, and complete positivity can be fully characterized. We expect these techniques-combining representation theory, matrix analysis and channel geometry-to extend naturally to infinite-dimensional settings and to other conjectures in quantum information theory involving symmetry and entanglement.

\begin{acknowledgments}
A.I. acknowledges financial support from the Spanish Ministry of
Economy and Competitiveness, through the Severo Ochoa Program for Centers of Excellence in RD
(SEV-2015/0554), the MINECO research project PID2024-160539NB-I00, and the Comunidad de
Madrid project TEC-2024/COM-84 QUITEMAD-CM.  
\end{acknowledgments}


\newcommand{\etalchar}[1]{$^{#1}$}

\end{document}